\documentclass[11pt]{extarticle}
\usepackage[english]{babel}
\usepackage{graphicx,charter}
\usepackage{authblk}
\usepackage{framed}
\usepackage[normalem]{ulem}
\usepackage{amsmath}
\usepackage[ruled,vlined,linesnumbered]{algorithm2e}
\usepackage{amsthm,thm-restate,thmtools}
\usepackage{amssymb}
\usepackage{amsfonts}
\usepackage{enumerate}
\usepackage{etoolbox}
\usepackage{subcaption}
\usepackage[utf8]{inputenc}
\usepackage[top=1 in,bottom=1in, left=1 in, right=1 in]{geometry}
\usepackage{xcolor,xfrac}
\usepackage{url}
\usepackage{enumitem}
\usepackage{tabularx}
\usepackage{algpseudocode}
\algdef{SE}[PROCEDURE]{Procedure}{EndProcedure}%
   [2]{\algorithmicprocedure\ \textproc{#1}\ifthenelse{\equal{#2}{}}{}{(#2)}}%
   {\algorithmicend\ \algorithmicprocedure}%

\DeclareMathOperator*{\argmax}{arg\,max}
\DeclareMathOperator*{\argmin}{arg\,min}

\renewcommand{\emptyset}{\varnothing}

\theoremstyle{definition}

\newtheorem{lemma}{Lemma}

\newtheorem{definition}{Definition}
\newtheorem{observation}{Observation}

\newcommand{\nameAs}{\text{students }}
\newcommand{\nameA}{\text{student }}
\newcommand{\symbA}{\mathcal{S}}
\newcommand{\A}{s}
\newcommand{\nameBs}{\text{colleges }}
\newcommand{\nameB}{\text{college }}
\newcommand{\symbB}{\mathcal{C}}
\newcommand{\B}{c}

\title{On Achieving Fairness and Stability in Many-to-One Matchings}
\author[1]{Shivika Narang}
\author[2]{Arpita Biswas}
\author[1]{Y Narahari}
\affil[1]{Indian Institute of Science (shivika, narahari @iisc.ac.in)}
\affil[2]{Harvard University (arpitabiswas@seas.harvard.edu)}
\date{}

\begin{document}
\maketitle
\begin{abstract}
\noindent The past few years have seen a surge of work on fairness in  allocation problems where items must be fairly divided among agents having individual preferences. In comparison, fairness in  settings with preferences on both sides, that is, where agents have to be matched to other agents, has received much less attention. Moreover, two-sided matching literature has largely focused on ordinal preferences. This paper initiates the study of fairness in stable many-to-one matchings under cardinal valuations.  Motivated by real-world settings, we study leximin optimality over stable many-to-one matchings. We first investigate matching problems with {\em ranked valuations\/} where  all agents on each side have the same preference orders or rankings over the agents on the other side (but not necessarily the same valuations). Here, we provide a complete characterisation of the space of stable matchings. This leads to {\bf FaSt}, a novel and efficient algorithm to compute a leximin optimal stable matching under ranked \textit{isometric} valuations (where, for each pair of agents, the valuation of one agent for the other is the same). Building upon FaSt, we  present an efficient algorithm, {\bf FaSt-Gen}, that finds the leximin optimal stable matching for a more general ranked setting. We next establish that, in the absence of rankings and under strict preferences, finding a leximin optimal stable matching is NP-Hard. Further, with weak rankings, the  problem is strongly NP-Hard, even under isometric valuations.  In fact, when additivity and non-negativity are the only assumptions, we show that, unless P=NP, no efficient polynomial factor approximation is  possible.
\end{abstract}
\section{Introduction}


Fairness is a rather desirable objective in most practical situations. In the past decade, the computational problem of achieving fairness has been receiving intense attention \cite{budish2011combinatorial,bouveret2016characterizing,plaut2020almost,chen2020fairness,caragiannis2019stable,barman2019fair,barman2021sublinear}. Several fairness notions have been studied for fair allocation problems, where agents have preferences over a set of items which must be fairly distributed (allocated) among the agents. Here, items don't have preferences over the agent they are allocated to. In contrast, two-sided (bipartite) matching problems \cite{gale1962college, bogomolnaia2004random, caragiannis2019stable, narang2020study} consist of two groups of agents where agents in each group have preferences over the agents in the other group. The task is to match the agents of one group to the agents of another group. In such a setting, stability and fairness are clearly two very natural and desirable objectives. 

The various fairness notions studied in the fair allocation literature have been mostly  unexplored for two-sided matching settings, with the exception of some very recent work \cite{freemantwo,gollapudi2020almost} which does not take into account stability.
Also, the two-sided matching literature has largely assumed ordinal preferences. This paper initiates the study of achieving fairness in stable, two-sided matchings under cardinal valuations. Motivated by real-world situations, we choose leximin optimality as the notion of fairness and focus on stable many-to-one matchings. We derive several novel algorithmic and complexity results for various special cases of many-to-one matchings. We first present a motivating example and lay out the space of problems explored in this paper. 

\subsubsection*{Motivating Example}
We consider the case of engineering admissions in India, where, admissions to undergraduate engineering colleges are based on a centralized exam called the Joint Entrance Exam (JEE). In 2021, over 2 million students gave the JEE \cite{IITJEEStat}. The students, based on their performance, are awarded an All India Rank (AIR). The AIRs determine the order in which they get to choose the college they wish to join, via a centralized admissions process. In turn, colleges are ranked by their reputation, determined by factors such as employment secured by the alumni, research output, brand, etc. For students, the reputation of a college indicates how much being an alumnus of this college will further their future prospects. This ranking of colleges is the same for nearly all students (except for a few who may have a strong location preference), and is influenced by rankings produced by national newspapers and magazines. While we cite the Indian setting, central ranking procedures for college admissions are used by several countries including Brazil, Germany, Taiwan, and the UK \cite{machado2016centralized}.  

Stability is a key criterion for college admissions. Informally, stability in this context requires that there should not exist a college-student pair where both can benefit by deviating from the given matching by matching with each other. 
This is clearly a very natural criterion and has been shown to be effective. 
In fact, Baswana et al \cite{Baswana2019india} found that using a stable matching mechanism for engineering admissions in India eliminates some of the glaring inefficiencies that were observed when stability was not considered.

In practice, besides stability, fairness must also be satisfied. Newer colleges are invariably ranked lower than  well established colleges. As a result, despite the newer colleges having adequate capacity and competitive quality, they are ignored by students, particularly top ranking students. Matchings where the large majority of students are matched to well established (higher ranked) colleges may be stable but are clearly unfair to newer colleges. This results in poor reputation and a loss of opportunity for the newer colleges. Consequently, these colleges continue to be poorly ranked in future years as well. Hence, without some explicit intervention, this further increases the discrepancy among the colleges for the next set of students in following years. 

In essence, this creates a vicious self-reinforcing hierarchy. To offset this discrepancy, a suitable notion of fairness could be implemented in the centralized admission process to these colleges. If newer colleges get more students, which they deserve, they would have an opportunity to excel and improve their rankings. Naturally, in our attempt to help these colleges, we must not ignore the interests of the students. That is, in order to help colleges, we cannot be indiscriminate with the students and send a majority of ranking students to low ranking colleges. 
Fairness must be maintained for the students and the colleges, while also being efficient.

\subsubsection*{Ranked Valuations} Several practical many-to-one matching settings belong to the {\em ranked valuation\/} setting where there are inherent rankings across the agents on each side. The class of ranked valuations has been well studied in the context of fair division (under various names) \cite{amanatidis2017approximation,barman2020approximation,bouveret2016characterizing}. This class is a generalization of identical valuations, a very well studied class of valuations \cite{plaut2020almost,barman2018greedy,BarmanS20,chen2020fairness}.   In ranked valuations, for each side, all agents have the same preference orders or rankings over the agents on the other side (but may not have the same valuations). For example, in labour market settings, workers are ranked based on their experience, while employers may be ranked on the wages they offer.  
In the example being discussed,  the ranking of colleges is the same for nearly all students and the colleges have a single ranking over the students (based on AIRs) (except for a few who have a strong location preference). This ranking is usually influenced by ranking surveys and word-of-mouth communication. These rankings lead to a structure over the space of stable matchings.

\subsubsection*{Isometric Valuations}
We also investigate  the isometric valuations setting, where the valuation of a \nameA for a \nameB is the same as that of the \nameB for the student. This assumption captures several practical settings.  
Students in India choose to pursue technical degrees, in large part, with the expectation of a well paying job upon graduation. The jobs that these students secure, in turn, determine the reputation of their colleges, especially for newer colleges. Consequently, in this setting, the value that a student and a college gain from each other is the expected quality of future opportunity awaiting the student. Further, the salaries students receive are also influenced by the reputation of the colleges they have attended. As a result, the rankings play an important role. 
Finding fair and stable matchings is practically relevant even with isometric valuations.

\subsubsection*{Leximin Optimality}
While there is a wide variety of fairness notions prevalent in the literature,  given that they focus on allocation settings, the large majority of these need not always maintain fairness across the two sets of agents. The question of whether agents on both sides must be treated equally is up for debate, but in this paper we assume that agents from either side are equally important. To this end, we focus on a fairness notion called \textit{leximin optimality} ~\cite{bezakova2005allocating, plaut2020almost,flaniganfair,flanigan2021fair} as our notion of fairness. Along with stability, this ensures a certain degree  of efficiency.
Informally, a leximin optimal matching is  one that maximizes the value of the worst-off agent (i.e. satisfies the maximin/santa claus criteria), and out of the matchings that achieve this, maximizes the value of the second worst-off agent, and so on. Essentially, this minimizes the discrepancy in the values achieved by each of the agents on either side, ensuring a fair balance in the interests of both sides.  For college admissions, leximin optimality improves values achieved by lower ranked students and  colleges with minimal compromise in  the matchings of higher ranked colleges and students. 

Another appealing aspect of the leximin notion is that it is an optimization based fairness notion with a guaranteed optimal outcome over stable matchings. While notions like egalitarian welfare (maximin welfare) and Nash social welfare are also guaranteed to exist, they are not sufficient for ensuring fairness for agents on both sides. Other popular fairness notions like envy-freeness and equitability and their relaxations  need not coexist with stability. As a result, leximin turns out to be an appealing fairness notion for two-sided matchings. We defer a comprehensive discussion on how these  fairness notions interact with stability to Appendix \ref{sec:envy}. 

\subsubsection*{Contributions}
In this paper, we seek to find a leximin optimal stable many-to-one matching with cardinal valuations. 
Table \ref{tab:contr} summarises our contributions.  Here, $m$  and $n$ are the number of \nameBs and students. Strict preferences mean that no agent is indifferent about any agent on the other side (each agent's valuations imply a strict linear order over agents on the other side).  Weak preferences allow agents to be indifferent between two agents (their valuations imply a partial order on the agents on the other side). 
\begin{table}[ht]
    \centering
    {\small
    \begin{tabular}{|p{0.32\linewidth}|p{0.27\linewidth}|p{0.28\linewidth}|}
        \hline
        \textbf{Nature of Valuations}    & \textbf{Strict Preferences}  & \textbf{Weak Preferences}\\
            \hline
        
        Ranked + Isometric & $O(mn)$ (Thm \ref{thm:main}) & Strongly NP-Hard (Thm \ref{lem:lexhard})\\
        
        Ranked Valuations & $O(m^2n^2) $ (Thm \ref{thm:algleft})& APX-Hard (Thm \ref{thm:apxhard})\\
        
        General & NP-Hard (Thm \ref{lem:stricthard}) & APX-Hard (Thm \ref{thm:apxhard}) \\
        General with $m=2$ & $O(n^2)$ (Thm \ref{thm:algconst}) & NP-Hard (Thm \ref{lem:lexhard}) \\
        \hline
        
    \end{tabular}}
    \caption{Summary of Results}
    \label{tab:contr}
    \vspace{-5mm}
\end{table}

Our algorithms all have a similar pattern, which only requires a way of iterating over the space of stable matchings. We believe this introduces a new approach which can inform future work on optimizing over the space of matchings.

\section{Preliminaries and Main Results}
\label{sec:prelim}

We first set up our model,  give the necessary definitions,  discuss the relevant work, and state the main results. 
\subsection{Definitions and Notation}
Let $\symbA = \{\A_1, \ldots, \A_n\}$ and $\symbB = \{\B_1, \ldots, \B_m\}$ be non-empty, finite, and ordered sets of \nameAs  and colleges, respectively. We assume that there are at least as many \nameAs as colleges, that is, $n\geq m$. We shall assume that each \nameB has a capacity (or budget) $b_j\leq n$ on the maximum number of \nameAs that can be matched to it.

Let $u_{i}(\cdot)$ and $v_{j}(\cdot)$ be the valuation functions of \nameA $\A_i$, and \nameB  $\B_j$, $i\in [n],\,j\in [m]$. This paper studies only non-negative and additive (for colleges) valuation functions. We define $U=(u_1,\cdots,\, u_n)$ and $V=(v_1,\cdots,v_m)$ and $B=(b_1,\cdots, b_n)$. Hence, an instance of stable many-to-one matchings (SMO) is captured by the tuple $I=\langle \symbA,\symbB, U$, $ V, B \rangle$. The ordering implied by an agent's valuations over agents from the other set are called preferences. Strict preferences require that the preference order is strict, or in other words, there for any agent, there are no two agents in the other set for whom they have equal value. Weak preferences imply that there may be ties in the preference order.

We shall also look at isometric valuations where the value that a student $s_i$ has for a college $c_j$ is the same as $c_j$'s value for $s_i$, that is,  $u_i(\B_j)=v_j(\A_i)$, for all $i\in [n]$ and all $j\in [m]$. Each instance with isometric valuations can be captured by the tuple $I=\langle \symbA,\symbB, V,B\rangle$. Here $V$ can equivalently represent an $n\times m$ matrix where $V_{ij}$ is the valuation of $\A_i$ for matching with $\B_j$, or $V_{ij}=u_{i}(\B_j)=v_{j}(\A_i)$. 

Rankings or ranked valuations imply that for all students, while the exact values need not be the same, the (strict) ordering over colleges is the same and analogously, each college has the same preference order over students. To this end, for convenience we assume that whenever the valuations are ranked, $u_{i}(\B_1)>u_{i}(\B_2)>\cdots >u_{i}(\B_m)$ for all $i\in [n]$ and $v_{j}(\A_1)>v_{j}(\A_2)>\cdots> v_{j}(\A_n)$ for all $j\in [m]$. Weak rankings shall allow for ties, $u_{i}(\B_1)\geq u_{i}(\B_2)\geq \cdots \geq u_{i}(\B_m)$ for all $i\in [n]$ and $v_{j}(\A_1)\geq v_{j}(\A_2)\geq \cdots\geq  v_{j}(\A_n)$ for all $j\in [m]$.. Thus, for ranked isometric valuations, each row and column of $V$ are sorted in decreasing order.

Our goal is to find a many-to-one matching $\mu$ of the bipartite graph $G=(\symbA, \symbB, \symbA\times\symbB)$ such that $\mu$ satisfies stability as well as fairness properties. A matching $\mu \subseteq \symbA\times\symbB$ is a subset of $\symbA \times \symbB$ such that each \nameA has at most one incident edge present in the matching, and the number of incident edges on a \nameB is \textit{at most} their capacity $b_j$, i.e., for any $i\in [n]$ and $j\in [m],\,\, |\mu(\A_i)|\leq 1$ and $|\mu(\B_j)|\leq b_j$. 

Alternatively, matching $\mu$ can be equivalently defined as a function that maps an agent to a set of agents, $\mu:\symbA\cup \symbB \mapsto 2^{\symbA}\cup \symbB$ such that for each \nameA $\A_i\in \symbA$, the function satisfies $\mu(\A_i)\in \symbB$ and $|\mu(\A_i)|\leq 1$, and for each \nameB $\B_j\in\symbB$, the function satisfies $\mu(\B_j)\subseteq \symbA$ and $|\mu(\B_j)|\leq b_j$ s.t. for each $i\in [n]$, $\A_i\in\mu(\mu(\A_i))$ and for each $j\in [m]$, $\B_j=\mu(p)$ for all $p\in \mu(\B_j)$. The valuation of a \nameA $\A_i$ under a matching $\mu$ is written as $u_i(\mu)=u_i(\mu(\A_i))$. For a \nameB $\B_j$ the valuation under matching $\mu$ is defined as $v_{j}(\mu)=\sum_{i\in \mu(c_j)}\ v_{j}(i)\geq 0.$ 
Next, we define two desirable properties of a matching, namely, stability and leximin optimality. 
\begin{definition}[Stable Matching]
A matching $\mu$ of instance $I=\langle \symbA,\symbB, V\rangle$ is said to be stable if no $(\A_i,\B_j)$ is a {\em blocking pair} for $\mu$.
\end{definition}

\begin{definition}[Blocking Pair]
Given a matching $\mu$, $(\A_i,\B_j)$ are called a {\em blocking pair} if $\A_i\notin \mu(\B_j)$ and there exists $\A_{i'}\in \mu(\B_j)$, $\B_{j'}=\mu(\A_i)$ such that $v_j(\A_i)>v_j(\A_{i'})$ and $u_i(\B_j)>u_{i}(\B_{j'})$. 
That is $\A_i$ and $\B_j$ prefer each other to (one of) their partners under $\mu$.
\end{definition}
\noindent Note that this definition of stability has also been called pairwise stability and justified envy-freeness in prior work. Stability often requires a non-wastefulness constraint as well. However, prior work \cite{zhang2018strategyproof,yahiro2018strategyproof} has found that non-wastefulness and pairwise stability are incompatible under distributional constraints, aimed at ensuring fairness. Thus, we do not demand our stable matchings to be non-wasteful. We denote the space of complete stable matchings, matchings which are stable and do not leave any agent unmatched, as $\mathcal{S}_C(I)$.

Our work aims to find the leximin optimal over the space of stable matchings. The most convenient way of identifying such a matching is using leximin tuples. The leximin tuple of any matching is simply the tuple containing the valuations of all the agents (\nameAs and colleges) under this matching, listed in non-decreasing order. Hence, the position of an agent's valuation in the leximin tuple may change under different matchings. The leximin tuple of a matching $\mu$ will be denoted by $\mathcal{L}_{\mu}$. The $t^{\text{th}}$ index of $\mathcal{L}_{\mu}$ is denoted by $\mathcal{L}_{\mu}[t]$.

\begin{definition}[Leximin Domination]
We say that matching $\mu_1$ leximin dominates $\mu_2$ if there exists a valid index $k$ such that $\mathcal{L}_{\mu_1}[k']=\mathcal{L}_{\mu_2}[k']$ for all $k'<k$ and $\mathcal{L}_{\mu_1}[k]>\mathcal{L}_{\mu_2}[k]$.
\end{definition}
\noindent We shall say that the leximin value of $\mu_1$ is greater than that of $\mu_2$ if $\mu_1$ leximin dominates $\mu_2$. This shall be denoted by $\mathcal{L}_{\mu_1} >\mathcal{L}_{\mu_2}.$
\begin{definition}[Leximin Optimal]
A leximin optimal matching $\mu^*$ is one that is not leximin dominated by another matching. 
\end{definition}
\noindent Essentially, we wish to find a matching that maximizes the left most value in the leximin tuple, of those that do, find the one that maximizes the second value and so on, which is NP-Hard, in general. Our goal is to find a {\em leximin optimal over all stable matchings}.

\subsection{Related Work}
Stable matchings and fairness have been studied somewhat independently for decades in the social choice literature. The formal study of the stable matching problem  began with Gale and Shapley's seminal paper \cite{gale1962college}. 
Their work initiated decades of research on both the theory \cite{roth1982economics,gale1985some,sethuraman2006many,caragiannis2019stable,narang2020study} and applications\cite{roth1989college,Baswana2019india,gonczarowski2019matching,roth2005kidney} of the stable matchings. This literature largely focuses on the structure of the stable matchings space\cite{roth1993stable,teo1998geometry}, incentive compatibility\cite{roth1982economics,narang2020study}, manipulating the Gale-Shapley mechanism \cite{teo2001gale,vaish2017manipulating} and deploying stable matchings in practice. Our work builds on this space by focusing simultaneously on fairness and stability in many-to-one matchings.
Important applications of stable matchings have initiated large bodies of work including college admissions, also known as school choice~\cite{roth1989college,biro2010college}, matching residents to hospitals~\cite{kavitha2004strongly,aldershof1996stable,klaus2005stable,irving2000hospitals,irving2003strong}  and kidney exchange\cite{roth2005kidney,roth2007efficient,biro2008three,irving2007cycle}. 
While this paper is inspired from the setting described in the introduction, it is relevant across the various applications of many-to-one matchings.

\subsubsection{Fairness in Matching Problems}
Fairness in matching settings  has often been defined from context-specific angles, such as college admissions \cite{zhang2018strategyproof,yahiro2018strategyproof,nguyen2019stable,lien2017ex}, focusing only on the colleges and not on the students. Our work looks at fairness from a broader perspective.  Some prior literature has focused on combating the inherent bias towards the proposing side in the Gale Shapley algorithm \cite{sethuraman2006many,klaus2009fair,huang2016fair,brilliantova2022fair}. There  is also some work on procedural fairness of the matching algorithms \cite{klaus2006procedurally,tziavelis2020fair}. 

However, these notions of fairness can be seen as group fairness notions and consider the two groups of agents, rather each agent individually. Further, matching literature has almost exclusively considered settings with ordinal preferences, and looks for ordinal fairness notions. We consider cardinal valuations and adopt \textit{leximin optimality} from {fair allocation} literature. Some very recent work has looked at envy-based fairness in many-to-many matchings \cite{freemantwo, gollapudi2020almost}, in contrast to our many-to-one matching setting, and does not require stability.
It is also important to note that the notion of justified envy-freeness studied in prior work \cite{wu2018lattice,aziz2019matching,yokoi2020envy} is different from the envy-freeness notion studied in allocation literature. Justified envy-free means no blocking pairs, and is also called pairwise stability or just, stability, as in this paper.
%
To avoid confusion with the fair allocation concept, we define stability as the absence of blocking pairs. 

\subsubsection*{Fairness in Allocation Problems}
 Fair allocation  refers to the problem of fairly dividing a set of items among a set of agents with known cardinal valuations. The fairness notions in allocation problems, unlike matching problems, are defined with respect to the agents only (and not the items). In this work, we adopt leximin optimality to define a two-sided leximin optimal in the matching scenario. 
A leximin optimal solution always exists and under, additive valuations, is Pareto Optimal.  The hardness of finding the leximin optimal allocation was established in \cite{bezakova2005allocating, plaut2020almost}. To the best of our knowledge, approximations to leximin have not been defined. Bezakova and Dani \cite{bezakova2005allocating} establish the APX-Hardness, for any factor better than $\sfrac{1}{2}$, of maximizing egalitarian welfare, or the Santa Claus problem. Here we aim to maximize the minimum value of any agent. Clearly, the leximin optimal does this and hence the hardness result extends to leximin. 

An important assumption for our algorithms to be able to bypass the hardness is rankings. Prior fair division literature has called such preferences as same order preferences \cite{bouveret2016characterizing}, fully correlated valuation functions\cite{amanatidis2017approximation} or instances with such preferences as ordered instances \cite{barman2020approximation,garg2020improved}. Such preferences also generalize other well studied subclasses of preferences like identical valuations \cite{barman2018greedy,barman2020uniform} and single parameter environments \cite{barman2019fair}. To the best of our knowledge such settings have not previously been studied for leximin.

In special case of dichotomous valuations \cite{bogomolnaia2004random,kurokawa2015leximin}, leximin optimality can be achieved in polynomial time. 
Bogomolnaia and Moulin \cite{bogomolnaia2004random} show that the maximum weight matching is stable under dichotomous preferences and satisfies a variety of properties including leximin optimality. 
Recent results by Benabbou et al.~\cite{benabbou2020finding} and Chen and Liu~\cite{chen2020fairness} study the properties of a leximin optimal allocation for restricted settings. 
To the best of our knowledge, our work is the first one to consider leximin optimality for the matching problem, without restricting to dichotomous valuations. Leximin has been used to ensure fairness in sortition \cite{flaniganfair,flanigan2021fair} and shows  significant improvement on the systems currently in use. 

\subsection{Overview of Main Results}
The goal of this paper is to find a stable and leximin optimal  many-to-one matching. Leximin is intractable in general, so we first look at ranked valuations, where the space of stable matchings has an appealing structure that we can exploit.
\begin{restatable}{lemma}{rankstruct}\label{lem:struct}
Given an instance of ranked  valuations, a matching $\mu$ is stable if and only if, for all $j\in[m]$,  $\mu(\B_j)=\{\A_{w_j +1},\cdots, \A_{w_j+k_j}\}$ where $k_j=|\mu(\B_j)|$ and $ w_j=\sum_{t=1}^{j-1} k_t.$
\end{restatable}


Lemma~\ref{lem:struct} ensures that a stable solution for a ranked instance would necessarily match a contiguous set of \nameAs to each $\B_j$. We exploit this in the algorithm FaSt (Algorithm~\ref{alg:fast}), which runs in time $O(mn)$. Its correctness is established in the following theorem.

\begin{restatable}{theorem}{algcorrect}\label{thm:main}
FaSt (Algorithm \ref{alg:fast}) finds a leximin optimal stable matching for ranked isometric valuations in time $O(mn)$.
\end{restatable}

Building on this algorithm, we give an algorithm for general ranked valuations FaSt-Gen (Algorithm \ref{alg:left}) which runs in time $O(m^2n^2)$. We prove its correctness in the following theorem.

\begin{restatable}{theorem}{alggen}\label{thm:algleft}
FaSt-Gen (Algorithm \ref{alg:left}) finds a leximin optimal stable matching given an instance of general ranked valuations.
\end{restatable}

These algorithms assume that the capacity of each \nameB is $n-1$ for ease of presentation. If not, the algorithms require additional checks for capacity constraints during initialization and switches, but the efficiency remains the same. The complete algorithms are given in Appendix \ref{sec:cap}. We now consider a setting with only two colleges and strict preferences. We find that here a polynomial time algorithm to find the leximin optimal stable matching does exist.

\begin{restatable}{theorem}{algconst}\label{thm:algconst}
FaSt-Const (Algorithm \ref{alg:m2}) finds a leximin optimal matching given an instance with strict preferences and $m=2$.
\end{restatable}

Unfortunately, we cannot expect efficient algorithms for much more general settings. In Section \ref{sec:hard}, we find that under strict preferences, without an assumption of rankings, the problem of finding a leximin optimal stable matching is NP-Hard when $n=\Omega(m)$. 

\begin{restatable}{theorem}{stricthard}\label{lem:stricthard}
It is NP-Hard to find a leximin optimal stable matching under strict preferences.
\end{restatable} 

The hardness of finding the leximin optimal stable matching doesn't follow from the hardness of leximin in the allocations, as stable matchings under strict preferences can't capture all feasible allocations. Given that rankings in combination with strict preferences give rise to efficient algorithms, a natural question is whether rankings with weak preferences can lead to polynomial time algorithms. This is not the case.  We show that even on relaxing strict rankings to weak rankings, finding the leximin optimal stable matching under isometric valuations is intractable, even with a constant number of colleges. %

\begin{restatable}{theorem}{isohard}\label{lem:lexhard}
It is NP-Hard to find the leximin optimal stable matching under isometric valuations with weak rankings with $m=2$ and strongly NP-Hard with $n=3m$.
\end{restatable}

Clearly, this implies that it is NP-Hard to find the leximin optimal stable matching in general. However, in the absence of isometric valuations, the problem is harder still. For the general problem, we establish APX-Hardness. 

\begin{restatable}{theorem}{polyapxhard}\label{thm:apxhard}
Unless P=NP, for any $\delta>0$, $c\in \mathbb{Z}^+$ there is no $1/cn^\delta$-approximation algorithm to find a leximin optimal stable matching under unconstrained additive valuations.  
\end{restatable}

Before we discuss these impossibility results, we detail our algorithmic results, starting with the setting of ranked isometric valuations.

\section{Algorithmic Results}\label{sec:algs}
The algorithms in this paper all have a similar outline. All three start with the student optimal stable matching and then iterate over the space of stable matchings to look for leximin value improvements by moving students from colleges on the right of the leximin tuple to those on the left. Doing this efficiently requires some method of going from any one stable matching to any other, making small changes at a time and never violating stability.

In this paper, the necessary method of iteration relies on strict preferences and ranked valuations.  We first show how these properties ensure the requisite structure over the space of stable matchings.

Recall that for any $i \in [n]$ and $j\in [m]$, $u_{i}(\B_1)>\cdots \,>u_{i}(\B_m)$ and $v_{j}(\A_1)>\cdots \,>v_{j}(\A_n)$. These rankings give structure to the space of stable matchings. We find that for a matching to be stable, it must be in accordance with the rankings. 

\rankstruct*

\begin{proof}

\noindent We prove the forward implication by assuming $\mu$ to be a stable matching. Let $\mu$ match each \nameB $\B_j$ to $k_j$ students. 
We inductively prove that the required property holds. We first show that $\mu(\B_1)$ is the set of first $k_1$ \nameAs that is, $\A_1,\cdots,\A_{k_1}$. If $k_1=0$, this is trivially satisfied. We can thus assume that $k_1>0$. Suppose $\B_1$ is not matched to all of $\A_1,\cdots,\A_{k_1}$ under $\mu$. As a result, there exists some $i \in [k_1]$ such that $\A_{i} \notin \mu(\B_1)$ and that $\mu(\A_{i}) = \B_j$ for some $j>1$. Consequently, there must be an $i'>k_1$ such that $\A_{i'}\in \mu(\B_1)$. However, $u_{i}(\B_1)>u_{i}(\B_j)$ and $v_{1}(\A_i)>v_{1}(\A_{i'})$, by assumption. Thus, $(\A_i, \B_1)$ form a blocking pair, which contradicts the fact that $\mu$ is a stable matching.

We now assume that, for the stable matching $\mu$, the lemma holds for the first $t-1$ colleges, $t\geq 2$ i.e.,
$\mu(\B_j) = \{\A_{w_j+1},\cdots, \A_{w_j+k_j}\}$ where $w_j=\sum_{j'=1}^{j-1} k_{j'}$ for all $ j< t.$
 
We now show that the lemma is true for $\B_{t}$. Suppose not, then there exists $i$ such that $w_t<i\leq w_t+k_t$ and $\A_{i}\notin\mu(\B_t)$. This implies that, there must exist some  $i'>w_t+k_t$ such that $\A_{i'}\in \mu(\B_t)$. Let $\mu(\A_i)=\B_j, j>t$. Now, as the valuations are ranked, $u_{i}(\B_t)>u_{i}(\B_j)$ and $v_{t}(\A_i)>v_{t}(\A_{i'})$.
This implies that $(\A_i,\B_{t})$ form a blocking pair which contradicts the stability of $\mu$. 
This proves that if $\mu$ is a stable matching that matches $\B_j$ to $k_j$ \nameAs, then $\mu(\B_j)$ must be equal to $\{\A_{w_j +1},\cdots, \A_{w_j+k_j}\}$.

We now prove the reverse implication. Let $\mu$ be a matching which matches $\B_j$ to $k_j$ \nameAs and $\mu(\B_j) = \{\A_{w_j+1},\cdots, \A_{w_j+k_j}\} $ where $w_j=\sum_{j'=1}^{j-1} k_{j'}$ for all $ j\in [m]$. 

Fix $i\in [n]$, and let $\mu(\A_i)=\B_j$. If $j=1$, $\A_i$ clearly does not form a blocking pair with any \nameB as they are matched to their most preferred college. If $j>1$, then $A_i$ prefers $\B_1,\cdots, \B_{j-1}$ to $\B_j$. But, all these \nameBs prefer each of the \nameAs matched to them to $\A_i$, because of the ranking. Consequently, $\A_i$ does not form any blocking pairs. As a result, there are no blocking pairs in $\mu$ and it is a stable matching. 
\end{proof}

Note that the number of stable matchings under ranked instances is thus $\binom{n+1}{m-1}$, hence, brute force would not be efficient if the number of \nameBs isn't constant. We first look at a subset of matching instances, called \textit{ranked isometric valuations} as a stepping stone to a more general result. 

Lemma \ref{lem:struct} in conjunction with isometric valuations, provides a structure over the leximin values of stable matchings, under ranked valuations. This enables us to find the leximin optimal stable matching. We now list some observations about the structure of a leximin optimal stable matching $\mu^*$ over the set of stable matchings under ranked valuations.

\begin{observation}\label{obs:non-empty}
No agent is unmatched under a leximin optimal stable matching, that is, $\mu(a)\neq \emptyset$ for any $a\in \symbA\times \symbB$.
\end{observation}

\noindent This in conjunctions with Lemma \ref{lem:struct} leads to the following observation.  

\begin{observation}\label{obs:first-last}
Under ranked valuations, in the leximin optimal stable matching, $\A_1\in\mu^*(\B_1)$ and $\A_n\in\mu^*(\B_m)$.
\end{observation}

For other students, we are only able to guarantee the following as a consequence of Lemma \ref{lem:struct}.
    
\begin{observation}\label{obs:predecessor}
If $\A_i$ is matched to $\B_j$ then $\A_{i-1}$ must be matched to either $\B_j$ or $\B_{j-1}$. 
That is, $\A_i\in\mu^*(\B_j)\ \Rightarrow\  \A_{i-1}\in \Big(\mu^*(\B_j)\cup \mu^*(\B_{j-1})\Big)$.
\end{observation}

We now use these observations, first in the specific context of ranked isometric valuations.

\subsection{Leximin for Ranked Isometric Valuations }

For ranked isometric valuations we have two additional observations which are instrumental in the execution of our algorithm.

\begin{observation}\label{obs:order}
\noindent Under ranked isometric  valuations, for each $\mu\in \mathcal{S}_C(I)$, the values of the \nameAs will appear in order of their rank in the leximin tuple, that is, for any  $i<i'$, $u_{i'}(\mu)<u_{i}(\mu)$. 
\end{observation}

\begin{observation}\label{obs:fixed}
For any $\mu$, under isometric valuations, $v_{j}({\mu})\geq u_{i}(\mu)$ for each \nameB $\B_j$ and for all $\A_i\in\mu(\B_j)$.
\end{observation}

\noindent These observations are critical for the particularly design of the FaSt (Algorithm \ref{alg:fast}), which outputs a leximin optimal stable matching under ranked isometric valuations in time that is linear in the size of the input.

\subsubsection*{{FaSt}: An Algorithm to Find a Fair and Stable Matching}
We first present an $O(mn)$ time algorithm, called FaSt, to find a leximin \underline{fa}ir \underline{st}able matching under ranked isometric  valuations. Recall that for isometric valuations, $V$ also denotes an $n\times m$ matrix where $V_{ij}=u_i(\B_j)=v_j(\A_i)$. For ease of presentation, we shall assume that for any \nameB $\B_j$, $b_j=n-1$, that is we effectively assume no capacity constraints. This is not a particularly limiting assumption however. The full algorithm with unrestricted $b_j$s is given in Appendix \ref{sec:cap} and follows the same logic. The time complexity of both the versions of the algorithm is $O(mn)$. We now use the structure outlined in the previous subsection to give an $O(mn)$ time algorithm for finding the leximin optimal stable matching. 

In essence, the algorithm starts with the \nameA optimal complete stable matching and gradually finds leximin optimal matching by improving the valuations of the \nameBs  in increasing order of rank, keeping stability and non-zero valuations for all agents as an invariant. By Observation \ref{obs:predecessor} and \ref{obs:fixed} we can start with \nameA $\A_n$ and iteratively decide the matchings for higher ranked \nameAs.

The initial stable solution matches \nameA $\A_n$ to \nameB $\B_m$ and \nameA $\A_1$ to \nameB $\B_1$ (using Observation~\ref{obs:first-last}). Further, the first $n-m+1$ \nameAs are matched to $\B_1$, for the remaining \nameBs, $\mu(\B_j)=\A_{n-(m-j)}$ for each $j\geq 2$. The algorithm then systematically increases the number of \nameAs matched to the lowest ranked \nameB $\B_m$ till there is a decrease in the leximin value, at which point, we switch to the next lowest ranked \nameB and repeat. During this, the number of \nameAs matched to all other \nameBs except $\B_1$ remains fixed. The Demote algorithm (Algorithm \ref{alg:demote}) does this.

Observe that when executing Demote, the values of the students do not increase and those of the colleges, other than $c_1$ do not decrease. Let $s_i$ be the highest indexed student being moves from $c_{j-1}$ to $c_j$ in Demote. By Observation \ref{obs:order}, $s_i$ will be the leftmost  in the leximin tuple of the students being moved. Consequently, whether there is an increase in the leximin value depends on the values of $s_i$ and $c_j$ before and after the demote procedure is executed.  


We shall say that the matching of a \nameA $\A_i$ (or a \nameB $\B_j$) is \textit{fixed} if it will not change any further during the execution of FaSt. 
The algorithm stops when one of the two happens: either the bottom $m-1$ \nameBs: $\B_2,\cdots, \B_m$ get fixed, or the $\B_1$ is matched to $\A_1$ only. Note that, every time a \nameA is sent to a lower ranked college or demoted, no \nameA sees an increase in their valuation and no \nameB, other than $\B_1$, sees a decrease in their valuation. Hence, whenever a \nameA is demoted, their position in the leximin tuple either moves to the left or stays in the same position, and the position of the corresponding \nameB moves to the right. Due to Observation \ref{obs:order}, the valuations of the previously matched \nameAs are unaffected. Hence, the algorithm optimizes for one leximin position at a time, before moving on to the next.

Demote (Algorithm \ref{alg:demote}) is maintains the invariant of a complete stable matching. If we decide to send \nameA $\A_i$ to \nameB $\B_j$ when they are currently matched to $\B_{j-1}$, it is not enough to simply do this one step. Doing this alone will make $\B_{j-1}$ unmatched, which violates our invariant. As a result, we need to match \nameA $\A_{i-1}$ to $\B_{j-1}$, this must continue till we send $\B_1$'s lowest ranked matched \nameA to $\B_2$. For this to be feasible, $\B_1$ must be matched to at least $2$ students. If not, no transfers are possible and as a result, no further improvement can be made to the leximin tuple. We ensure this feasibility in the \texttt{while} loop condition in Step 7 of FaSt. 

\begin{algorithm}[!ht]{
    \KwIn{A matching $\mu$, \nameA index $i$, \nameB indices $down$ and $up$.}
    \KwOut{$\mu$}
    //Move $s_i$ to $c_{down}$ while reducing the number of students matched to $c_{up}$ while maintaining the number of students matched to all other colleges. 
    Set $t\gets i$\;
    Set $p\gets down$\;
    \While{$p>up$}{
        $\mu(\B_{p-1})\gets \mu(\B_{p-1})\backslash \{\A_t\}$\;
        $\mu(\B_{p})\gets \mu(\B_{p})\cup \{\A_t\}$\;
        $t\gets t-1$, and 
        $p \gets p-1$\;
    }}
    \caption{Demote}\label{alg:demote}
    \end{algorithm}

 \begin{algorithm}[!ht]
    {
     \KwIn{Instance of ranked isometric valuations $\langle \symbA,\symbB,V\rangle$}
     \KwOut{$\mu$}
     Initiate a stable matching: $\mu(\B_1)\gets \{\A_1,\cdots, \A_{n-m+1}\}$ and $\mu(\B_j)\gets \{\A_{n-(m-j)}\}$ for $j\geq 2$\;
     Initialize $i\gets n-1$, $j \gets m$ \;
     Set $\mathcal{L}$ as the leximin tuple for $\mu$\;
     Set $pos[i]$ as the position of $\A_i$ in $\mathcal{L}$, $i\in [n]$\;
     $//$tie breaking for position in $\mathcal{L}$: agents who attain the same value are listed in increasing order of rank, with \nameA before \nameBs.
     
     Initialize $\mathcal{F}\gets \{\A_n\};\,//$ stores the agents whose matching is fixed.\\
     \While{$i>j-1$ AND $j>1$ }{
     \eIf{$v_{j}(\mu)\geq V_{i(j-1)}$}
     {
        $j\gets j-1$;
     }
     {
     
      \eIf{$[V_{ij}>v_{j}(\mu)]$}{
       $\mu \gets Demote(\mu,i,j,1)$\;
       }{
      \eIf{$V_{ij}<v_{j}(\mu)$}{
        $j\gets j-1$\;
      }{
        // Look ahead: does sending $\A_i$ to $\B_j$ improve leximin\\
        $k\gets i$,\, 
        $t\gets pos[i]$, and
        $\mu'\gets \mu $\;
        \While{$k>j-1$}{
            \eIf{$V_{kj}>\mathcal{L}[t] $}{
                $i\gets k$ and  $\mu \gets Demote(\mu',k,j,1)$\;
                \textbf{break}\;
            }{
                \eIf{$V_{ij}<v_{j}(\mu)$}{
                    $j\gets j-1$\;
                    \textbf{break}\;
                }{
                    //Another tie, send $\A_k$ to $\B_j$ tentatively\\
                    $\mu'\gets Demote(\mu',k,j,1)$\;
                    $k\gets k-1$, and $t\gets t+1$\;
                }
      
            }
            }
            \If{$k=j-1$ AND $\mu\neq \mu'$}
                    {$j\gets j-1$\;}
      }
      }
     }
    
    $\mathcal{F}\gets \{\A_i,\cdots,\,\A_n\}\cup \{\B_{j+1},\cdots \B_m\}$\;
    $Update(\mathcal{L},\mu,pos)$\;
    $i\gets i-1$\;
    }
    }
     \caption{FaSt}
    \label{alg:fast}
    \end{algorithm}

\algcorrect*
    \begin{proof}
    We first prove that the correctness of FaSt.
    
\noindent \textbf{Correctness}: During initialization, FaSt matches $\A_n$ to $\B_m$, and fixes the match by assigning $\A_n$ to the set $\mathcal{F}$. This step indicates that $\A_n$ remains matched to $\B_m$ throughout the execution of the algorithm. Recall that, by Observation \ref{obs:order}, for the leximin optimal stable matching $\mu^*$, the first entry in its leximin representation $\mathcal{L}(\mu^*)$ is $V_{nm}=u_{n}(\B_m)$, which is ensured by FaSt.

We shall say that the matching of a \nameA $\A_i$ (or a \nameB $\B_j$) is \textit{fixed correctly} if the it is matched as in a leximin optimal matching. Let the matchings of $\A_{t+1},\cdots, \A_n$ be fixed, $t\leq n-1$ and let $\A_{t+1}$ be matched to $\B_d$. Hence, the matchings of $\B_{d+1},\cdots, \B_m$ are also fixed correctly. 
    
By Observation \ref{obs:order}, $\B_1,\cdots, \B_d$ and $\A_1,\cdots, \A_{t-1}$ must occur to the right of $\A_t$ in the leximin tuple of the leximin optimal stable matching. Similarly, $\A_{t+1},\cdots, \A_n$ must be listed to the left of $\A_t$. This will not change, irrespective of the way the other agents are matched, as long as stability is maintained. Also, since $\A_t$'s valuation for $\B_d$ does not depend on how any other agent is matched, it will not change once it is fixed. As a result, $\mu^*$ matches $\A_t$ to $\B_d$ if and only if it results in a leximin dominating matching.

By sending $\A_t$ to $\B_d$, $\A_t$ moves to the left of its current position in the leximin tuple. To improvement the leximin value, its value at the new position should not be lower than the current value there. Thus, if $V_{td}$ is less than $\B_d$'s current valuation, $\A_t$ must be fixed to $\B_{d-1}$, the matching of $\B_d$ must be matched as in the current matching. Even if we add more \nameAs to $\B_d$'s matching, $\A_t$ will remain at the same position, due to Lemma \ref{lem:struct} and as a result will be leximin inferior to the current matching. 
        
In the case where these two values are equal, we must look ahead to see whether by demoting $\A_t$ we eventually result in an increased leximin value. Simply comparing $\B_d$'s new valuation need not suffice. By demoting $\A_t$ we may land in a setting where $\B_d$'s new valuation is less than $V_{t(d-1)}$, resulting in a matching which is leximin inferior to the current. However by sending sending more \nameAs to $\B_d$, we may improve upon the current leximin tuple. Let the position of $\A_t$ when matched to $\B_d$ be $k$. Thus, we must compare the new leximin values (obtained by sending $\A_{t+1}$ to $\B_d$) at position $k+1$ with the current leximin value at position $k+1$. If the values are equal we may have to look ahead further and accordingly fix the matching. Thus, when the matchings of $\A_{t+1}$ and $\A_n$ are fixed we can correctly fix $\A_t$ as in a leximin optimal matching.
    
    \noindent \textbf{Termination}: The algorithm considers each \nameA- \nameB pair at most once. All computations  during one such considerations can be done in constant time. 
    The update to the leximin tuple $\mathcal{L}$ and the array $pos$ can be done in time $O(n+m)=O(n)$ as well. Thus, the time taken is $O(mn)$ in the worst case, which is linear in the number of edges of the underlying bipartite graph..
    \end{proof}
    
\noindent 
The success of this approach is contingent on i) the inherent rankings, ii) valuation functions of the \nameBs being additive, and iii) the five observations listed earlier. We do not rely  directly on the fact that the valuations are isometric. Consequently, whenever we have these three requirements met in a matchings instance, we can use FaSt as is to find the leximin optimal stable matching. Essentially, FaSt works correctly for the space of ranked instances where the valuations are additive and follow the following restriction: $u_{i}(\B_j)\leq v_{j}(\A_i) $ and $u_1(\B_j)\geq u_2(\B_j)\geq \cdots \geq u_n(\B_j)$ for all $i\in [n]$ and $j\in [m]$. The extension of FaSt for instances without the assumption that $b_j=n-1$ is detailed in Appendix \ref{sec:cap}.

\subsection{General Ranked Valuations}
\label{sec:beyond}

We now discuss the algorithm to find the leximin optimal stable matching under general ranked valuations where values across an edge of the bipartite graph need not be the same. For ease of presentation, we shall assume that the capacity of each $\B_j$, $b_j=n-1$. As in the case of isometric valuations this is not a particularly simplifying assumption.  The full algorithm without this assumption is given in Appendix \ref{sec:cap}. The time complexity of the two algorithms is the same.

We no longer assume any relation between the $u_i$s and $v_j$s. Consequently, Observations \ref{obs:order} and \ref{obs:fixed} stated in the start of the section, need not hold any longer. Hence, for the general ranked setting, the approach followed in Algorithm \ref{alg:fast} will no longer be applicable to find a leximin optimal stable matching, beyond the case when $m=2$.

The approach behind FaSt, essentially starts with the \nameA optimal stable matching and increases the number of \nameAs matched to the lowest ranked \nameB till there is a decrease in the leximin value. Then it moves to the next lowest ranked \nameB and repeats. Without Observation \ref{obs:fixed}, we may have that the \nameBs lie to the left of the \nameAs in the leximin tuple. As a result, we may continue to increase the valuation of the lowest ranked \nameB without much increase in the valuations of other \nameBs which may now lie to the left of $\B_m$. Consequently, we may end up only balancing the values of $\B_1$ and $\B_m$, not finding the leximin optimal. 

\subsection{FaSt-Gen}


For the general ranked valuations setting, we propose another approach that starts with the student optimal stable matching. In each iteration, we increase the number of \nameAs matched to the leftmost unfixed \nameB by one if it increases the leximin value. If not, the algorithm fixes the upper limit of this \nameB and the lower limit of the next highest ranked one. In order to do this, we decrease one \nameA from a higher ranked \nameB using the demote procedure (Algorithm \ref{alg:demote}). The choice of this higher ranked \nameB is $\B_1$ initially, till it can no longer give out any more students, then we consider the next \nameB whose lower limit is not fixed. Thus, the algorithm gradually fixes the upper and lower limits of the \nameAs matched to each college.

There is one more subtle feature to note. In the absence, of Observation \ref{obs:order}, when increasing the number of \nameAs of a particular \nameB a leximin decrease may happen due to a higher ranked \nameA- \nameB pair. That is, a \nameA which is being moved, but not to the lowest unfixed \nameB may cause a leximin decrease. In such cases, we temporarily fix or ``soft fix'' the upper limits of some \nameBs and then unfix them, once the \nameB which caused the leximin decrease is starts giving out \nameAs currently matched to it. We  assume a routine we call $sourceDec(\mu_1,\mu_2)$ which returns the agent who is the cause of the leximin decrease in $\mu_1$ vs $\mu_2$.
\begin{algorithm}[!ht]{
    \KwIn{Instance of general ranked valuations $\langle \symbA,\symbB,U,V\rangle$}
    \KwOut{$\mu$}
    //Initiate a stable matching:\; 
    $\mu(\B_1)\gets \{\A_1,\cdots, \A_{n-m+1}\}$ and $\mu(\B_j)\gets \{\A_{n-(m-j)}\}$ for $j\geq 2$\;
    $UpperFix \gets \{\B_1\},\,\, LowerFix \gets \{\B_m\}$\;
    $SoftFix \gets \emptyset$\; $Unfixed \gets UpperFix^c$\;
    
    \While{$|UpperFix \backslash LowerFix| + |LowerFix|<m$}{
        $up\gets \min_{j\notin LowerFix } j$\;
        $down \gets \argmin_{j\in Unfixed} v_j(\mu)$.\;
        $SoftFix \gets SoftFix \backslash \{(j,j')| j'\leq up < j \}$\;
        \eIf{$|\mu(\B_{up})|=1$ OR $\,v_{up}(\mu)\leq v_{down}(\mu)$}{
            $LowerFix\gets LowerFix \cup \{\B_{up}\}$\;
        }
        {
            $\mu' \gets Demote(\mu,down,up)$\;
            \eIf{$\mathcal{L}_{\mu'}\geq\mathcal{L}_{\mu}$}
            {
                $\mu \gets \mu'$\;
            }
            {//Decrease in leximin value, need to check the source of the decrease\;
            \eIf{$sourceDec(\mu',\mu)=\B_{up}$}
            { 
                $LowerFix\gets LowerFix \cup \{\B_{up}\}$ \label{eq:a}\;
                $UpperFix\gets UpperFix \cup \{\B_{up+1}\}$
            }{
                \eIf{$sourceDec(\mu',\mu)\in \symbA$}
                {
                    $\B_t\gets \mu(sourceDec(\mu',\mu))$\;
                    $LowerFix\gets LowerFix \cup \{\B_t\}$\;
                    $UpperFix\gets UpperFix \cup \{\B_{t+1}\}$\;
                    $A\gets \{j| j>t+1, j\in Unfixed \}$\;
                    $SoftFix \gets SoftFix \cup (A\times \{t+1\})$\;
                }
                {//The source of the decrease is a \nameB which is still unfixed. Need to check if this may lead to an eventual leximin increase
                    $(\mu,LowerFix,UpperFix,SoftFix) \gets LookAheadRoutine(\mu, down, LowerFix$, $UpperFix$, $SoftFix)$\;
                }
            }
            
            }
        }
        $Unfixed \gets \{j|j\notin UpperFix $ or $ (j,j')\notin SoftFix$ for some $j'>j\}$\;
        
    }}
    \caption{FaSt-Gen}\label{alg:left}
    \end{algorithm}

Note that the upper limit of a \nameB $\B_j$  is fixed permanently  when:

\begin{itemize}
    \item Lowest ranked \nameA matched to $\B_{j-1}$ causes a decrease in the leximin value on being matched $\B_j$ OR
    \item $up=j-1$ and $\B_{j-1}$ causes a leximin value decrease
\end{itemize}

\noindent This includes the case when $|\mu(\B_{up})|=1$. The lower limit of a \nameB $\B_j$ is fixed when: 

\begin{enumerate}
    \item $\B_{j+1}$'s upper limit is fixed
    \item $j=up$ and the valuation of $\B_j$ becomes less than that of the any unfixed \nameB, OR
    \item Giving out any more \nameAs would cause a leximin decrease, This includes the case when $|\mu(\B_{j})|=1$. 
\end{enumerate}
\noindent We now prove the correctness of the algorithm. 

\alggen*

\begin{proof}
We induct on $t\in [m]$, the rank of the \nameBs. For $j\in [m]$, let $l_j$ and $h_j$ indicate the ranks of the lowest and highest ranked \nameAs matched to $\B_j$.

\textbf{Base Case: t=1} When $t=1$, $h_1=1$, that is, $\A_1$ is the highest ranked \nameA matched to $\B_1$. This follows from Observation \ref{obs:first-last}.
Clearly, $\A_{(l_1)+1}=\A_{h_2}$. Recall that $\A_{h_2}$ was not included in $\mu(\B_1)$ as it would be leximin inferior to giving it out. For any more \nameAs to be included in $\mu(\B_1)$ would thus be leximin inferior, thus no leximin optimal matching would match $\B_1$ to more \nameAs.

Now $\B_1$'s lower limit was fixed because either: i) $\A_{l_1}$'s valuation for $\B_2$ would be too low, OR ii) $\B_1$ could no longer give out anymore \nameAs without causing a leximin decrease.
Consequently, any matching where $\A_{l_1}$ is not matched to $\B_1$  will be leximin inferior to the current , restricted to the values of the first two \nameBs and $\{\A_1,\cdots,\, \A_{l_2}\}$, as in all of those matchings, the valuations of $\A_{l_1}$ and $\B_1$ will always be strictly lower. Thus, $c_1$'s lower limit is fixed correctly. 

\textbf{Induction Hypothesis:} Let the matching of ${\B_1,\cdots,\, \B_{t-1}}$ be fixed as in an optimal matching. As a result, any matching which does not match ${\B_1,\cdots,\, \B_{t-1}}$ would be leximin inferior to those that do with respect to the matchings of the first $t$ \nameBs and $\{\A_1,\cdots,\,\A_{l_t}\}$.

Now consider $\B_t$. From Induction Hypothesis, $\B_t$'s upper limit is fixed correctly. There are now three possible cases: 

\textbf{Case 1}: $t=m$. In this case, $l_t=l_m=n$ which is clearly correct. 

\textbf{Case 2}: $t< m$ and at no point in the execution of the algorithm does $up=t$. In this case, $\B_t$'s lower limit has been fixed in Step \ref{eq:a} of FaSt-Gen. 
That is, $\A_{l_t}$'s valuation for $\B_{t+1}$ was too low. Thus any matching which matches $\A_{l_t}$ to a lower ranked \nameB will be leximin inferior to the current one with respect to the matchings of the first $t+1$ \nameBs and $\{\A_1,\cdots,\,\A_{l_{t+1}}\}$. 

\textbf{Case 3}: At some point in the execution $up=t$. Analogously to the base case, changing the matching of $\B_t$ decreases the leximin value.

\noindent Hence, the lower limit of $\B_t$ is as in an optimal matching.
\end{proof}
\begin{algorithm}[!ht]{
    \KwIn{$I,\mu,down,LowerFix,UpperFix,SoftFix$}
    \KwOut{$\mu,LowerFix,UpperFix,SoftFix$}
    $\langle \mu', LF, UF \rangle \gets \langle \mu,\, LowerFix,\, UpperFix\rangle $\;
    \While{$|LF|+|UF\backslash LF|<m$}{
        $up\gets \min_{j\notin LowerFix } j$\;
        \eIf{$|\mu(\B_{up})|=1$ OR $v_{up}(\mu)\leq v_{down}(\mu)$}{
            $LF\gets LF \cup \{\B_{up}\}$\;
        } 
        {
            $\mu' \gets Demote(\mu',up,down)$\;
            \eIf{$\mathcal{L}(\mu')\geq\mathcal{L}(\mu)$}
            {
                $\mu \gets \mu'$, 
                $LoweFix\gets LF$, 
                $UpperFix \gets UF$\;
                \textbf{break}\;
            }
            {//Decrease in leximin value, need to check the source\;
            \eIf{$sourceDec(\mu',\mu)=\B_{up}$}
            { 
                $LF\gets LF \cup \{\B_{up}\}$, $UF \gets UF \cup \{\B_{up+1}\}$\;
            }{
                \If{$sourceDec(\mu',\mu)\in \symbA$}
                {
                    $\B_t\gets \mu'(sourceDec(\mu',\mu))$\;
                    \eIf{$t=down$}{
                        $UpperFix\gets UpperFix \cup \B_{down}$\;
                    }
                    {
                         $SoftFix \gets SoftFix \cup (down,t)$\;
                    }
                    \textbf{break}\;
                    
                }
                
            }
            
            }
            
            }
    }}
    \caption{LookAhead Routine}\label{alg:genlook}
    \end{algorithm}
\noindent \textbf{Look Ahead Routine:} 
As with isometric valuations, it may be possible for a decrease in the leximin value to be made up for. We may encounter a setting where the valuation of a demoted \nameA decreases to the  former valuation of $\B_{down}$. The valuation of $\B_{down},$ while it does increase, it is less than the \nameA's former valuation, causing a leximin decrease. We must decide if it is possible for $\B_{down}$ to make up for this decrease in valuation. Hence, we must look ahead till  one of the following occurs: i) There is a leximin increase, ii) There is a leximin decrease due to a \nameA OR iii) No more \nameAs can possibly be added to $\B_{down}$'s matching.

In the first case, we must proceed with the new matching, and in the last we must revert to the old one, and fix $\B_{down}$'s upper limit. In the case that there is a leximin decrease due to a \nameA being matched to $\B_{down}$, then clearly, we must revert to the old matching and fix $\B_{down}$'s upper limit. The last possibility is that there is a leximin decrease due to a \nameA matched to a higher ranked \nameB. In this case, we have that so far our current matching continues to be leximin inferior to the old one, and for now we cannot add any more \nameAs to $\B_{down}$. Thus, as in the case when this happens outside of a look ahead routine, we temporarily fix $\B_{down}$ till the offending \nameB is open for giving out \nameAs.

 \textbf{Time Complexity}: Without Observations \ref{obs:order} and \ref{obs:fixed}, each leximin comparison takes $O(m+n)=O(n)$ time. Updating a matching takes $O(m)=O(n)$ time. Whenever a particular \nameA - \nameB is considered, a leximin comparison of the updated matching must be done. Pair $(\A_i,\B_j)$ may be considered  at most $j-1\leq m-1$ times. Thus the algorithm considers $mn$ pairs at most $m$ times each taking $O(n)$ time. Consequently, the time complexity is $O(m^2n^2)$.  
Even in the presence of capacity constraints, it is possible to find the leximin optimal stable matching in time $O(m^2n^2)$. The details of the algorithm are given in Appendix \ref{sec:cap}.

\subsection{Strict Preferences and Constant Number of Colleges $m=2$}
We now show that when all agents' preferences are strict (but without the assumption of rankings) and there are exactly two colleges, even with an unrestricted number of students, we have a polynomial time algorithm to find the leximin optimal stable matching. Observe that, in this setting, when $m=2$, any student $s_i$ only forms a blocking pair with her most preferred college. As a result, for a matching $\mu$ to be stable, we need that for each student $s_i$, either $s_i$ is matched to their favourite college (say $c_j$) or $c_j$ prefers all the students matched to it under $\mu$ over $s_i$.
A simple corollary of this is that the matching where all students are matched to their preferred college is stable. This is the student optimal stable matching. 


To this end, let us set up some notation for this specific case of $m=2$. 
We use $c_{-j}$ to denote the college other than $c_j$. We shall use the function $\alpha$ to capture the students' preference relations. As $m=2$, we need to only know a student's most preferred college to determine their preference relation. To this end, we use $\alpha (s_i)$ to denote $s_i$'s (most) preferred college and $-\alpha(s_i)$ to denote the college other than $\alpha(s_i)$. For $j=1,2$, we shall define the sets $A_j=\{ s_i|\alpha (s_i)=c_j\}$. Observe that $A_1$ and $A_2$ partition the set of students, and hence the complement of $A_1$, i.e. $\overline{A}_1=A_2$ and similarly, $\overline{A}_2=A_1$. For a matching $\mu$, let $k_j(\mu)$ denote the number of students from $A_j$ matched to $c_j$ under $\mu$, i.e., $k_j(\mu)=\mu(c_j)\cap A_j$. Now let $l_j(\mu)$ denote the number of the rest, i.e., $l_j(\mu)=|\mu(c_j) \setminus A_j|=|\mu(c_j)|-k_j$. Whenever, the matching being referred to is clear, we shall drop the $(\mu)$.

We shall use the functions $top_j$ and $bottom_j$ to take as input a set of students and a number $k$ and return college $c_j$'s most preferred and least preferred $k$ students from that set, respectively. Whenever $k>1$, we shall use it as $top_j(S,,k)$, and when $k=1$, we shall simply write $top_j(S)$. Now for a matching $\mu$ to stable, for every $s_i \in A_j\setminus \mu(c_j)$, that is every student who prefers $c_j$ but is not matched to it under $\mu$, no student that $c_j$ prefers less than $s_i$ should be matched to it under $\mu$. This means that, the students from $A_j\cup\mu(c_j)=top_j(A_j,k_j(\mu))$. Similarly, those not in $A_j$ must be the $l_j$ least preferred students of $c_{-j}$, that is, $\mu(c_j)\setminus A_j = bottom_{-j}(\overline{A}_j,l_j(\mu))$. Further, out of all these students from $\overline{A}_j$, $c_j$'s least preferred one should still give $c_j$ more value than the most preferred student from $A_j$ matched to  $c_{-j}$. otherwise, there will be a blocking pair. This is captured in the following observation. 

\begin{observation}\label{obs:m2stable}
Given matchings instance $I=<\symbA,\symbB,U,V>$ s.t. $m=2$, a matching $\mu$ is stable if and only if the following hold. For each $j=1,2$, $\mu(c_j)=top_j(A_j,k_j)\cup bottom_{-j}(\overline{A}_{j},l_j)$. Further, it also must hold that if for $c_j$, $l_j>0$, then $$v_j(top_j(A_j \setminus \mu(c_j))))<v_j (bottom_j( \mu(c_j) \setminus A_j)).$$
\end{observation}


We can now use this observation to show a method of going from any one stable matching to another, by toggling the matching of one student at a time and none of the intermediate matchings violating stability. This can be achieved as follows: Let the two matchings be $\mu_1$ and $\mu_2$. For each $j\in \{1,2\}$ s.t. $k_j(\mu_1)<k_j(\mu_2)$, toggle the matching of the students matched to $c_j$ under $\mu_2$ but not $\mu_1$ in decreasing order of $c_j$'s preference. That is, we toggle the matchings of students in $(\mu_2(c_j) \setminus \mu_1(c_j))\cap A_j$. Observe that this does not violate stability. After this is done, for each college $c_j$ such that it is matched to more students from $\overline{A}_j$ under $\mu_2$ than $\mu_1$, that is, $l_j(\mu_1)<l_j(\mu_2)$, bring in the students in $(\mu_2(c_j) \setminus \mu_1(c_j))\cap \overline{A}_j$ in increasing order of $c_{-j}$'s preference. The resultant matching is $\mu_2$. Consequently, we can go from any one stable matching to another without violating stability. 

We shall make use of this structure to find the leximin optimal stable matching. Our algorithm will proceed as follows: we first match all students to their preferred college. We then transfer out students from the rightmost college on the leximin tuple to the leftmost till we encounter one of the three possible settings: 

\begin{enumerate}
    \item leximin decrease due to a student
    \item violation in stability
    \item needing to undo a previous transfer
\end{enumerate}

All three of these conditions will be captured by a set $\mathcal{F}$ of forbidden pairs, where any student college pair whose being matched would cause any one of these three conditions are stored. We assume a function that toggles the matching of a student named $Toggle$ and it takes as input the student and the matching. We shall use $\mu_P$ to store the matching with the best leximin value discovered so far. We can use this as an alternate to having a separate lookahead procedure, as there are exactly 2 colleges. 

\begin{algorithm}[!ht]{
    \KwIn{Instance $\langle \symbA,\symbB,U,V\rangle$ with $m=2$ }
    \KwOut{$\mu_P$}
    $\mu(s_i)\gets \alpha (s_i)$ for each $i\in [n]${\footnotesize \Comment{Initialize to a stable matching}}\;
    $\ell  \gets \argmin_{j=1,2} v_j(\mu)$, $r \gets -\ell $\;
    $\mathcal{F}\gets \{(s_i,-\alpha (s_i))| u_i(-\alpha (s_i))<v_{\ell}(\mu)\}${\footnotesize \Comment{Prevents irreversible leximin decrease}}\;
    $\mu_P \gets \mu$\;
    \While{$(bottom_{r}(\mu(r)),\ell )\notin \mathcal{F}$ }{
        $i^* \gets bottom_{r}(\mu(r))$, $j^*\gets r$\;
        $\mu \gets Toggle(\mu,s_{i^*})$\;
        $\mu_P \gets \argmax \{\mathcal{L}_{\mu},\mathcal{L}_{\mu_P}\}$\;
        $\ell \gets \argmin_{j=1,2} v_j(\mu)$, $r \gets -\ell $\;
        $\mathcal{F}\gets \mathcal{F} \cup \{(s_i,-\alpha(s_i))| u_i(-\alpha(s_i))<v_{\ell }(\mu) \}$ {\footnotesize \Comment{Prevents irreversible leximin decrease}}\\
        \quad \quad \quad \,$\,\cup \{(s_i,c_{j^*})| v_{j^*}(s_{i^*})\leq v_{j^*}(s_{i^*})\}${\footnotesize \Comment{Prevents stability violation and undoing this iteration}}
    }
    }
    \caption{FaSt-Const}\label{alg:m2}
    \end{algorithm}

Note that students are moved from the right most college to the leftmost. Doing the opposite, will not lead to an increase in leximin value at any time. We now show that the student-college pairs that are forbidden are never matched in a leximin optimal stable matching. 

\begin{lemma}\label{lem:forbid}
Given  a matchings instance $I=\langle \symbA,\symbB,U,V\rangle$ with $m=2$, let $\mathcal{F}$ be as in the end of Algorithm \ref{alg:m2} when run on $I$. A leximin optimal stable matching on $I$ does not match any $s_i,c_j$, s.t. $(s_i,c_j)\in \mathcal{F}$.
\end{lemma}

\begin{proof}
Let $(s_i,c_j)\in \mathcal{F}$ and let $\mu$ be the matching when $(s_i,c_j)$ are {\em first} added to $\mathcal{F}$. Three possible reasons are:  

\begin{enumerate}
    \item[(i)] $u_i(-\alpha (s_i))<\min_{j=1,2} v_j(\mu)$.
    
    This is only possible (for the first time when $(s_i,c_j)$ are added to $\mathcal{F})$ if $\mu(s_i)=\alpha(s_i)$ and $c_j=-\alpha(s_i)$. Now observe that as matching $s_i$ to $-\alpha (s_i)$, will lead to a decrease in the value of $s_i$, the left most of the three agents whose values are changing, this will cause a leximin decrease. As a result, any matching $\mu'$ where $u_i(-\alpha (s_i))<\min_{j=1,2} v_j(\mu')$, the leximin value will decrease by toggling $s_i$'s matching to $-\alpha (s_i)$. The only way to undo this decrease, is to ensure that $s_i$ is matched to $\alpha (s_i)$. 
    
    \item[(ii)] $s_i$ is being moved from $\alpha (s_i)$ to $-\alpha (s_i)$. 
    
    Here, $c_j=\alpha (s_i)$, and as $s_i$ is being moved to $-\alpha (s_i)$. This implies that any other stable matching where $s_i$ is matched to $\alpha (s_i)$ has either been explored or involves moving students from $-\alpha(s_i)$ (which is the leftmost college in the leximin tuple) to $\alpha(s_i)$. As a result, none of these matchings will be leximin superior to $\mu_P$ which stores the matching with the best leximin value out of the matchings explored. Consequently, moving $s_i$ back to $\alpha (s_i)$, even later in the execution of the algorithm would either cause a stability violation, or undo multiple steps in the execution of the algorithm, none of which will lead to a higher leximin value than $\mu_P$.
    
    \item[(iii)] Another student $s_{i'}$ is being moved from $\alpha (s_{i'})$ to $-\alpha (s_{i'})$ and $v_{\alpha(s_{i'})}(s_i)<v_{\alpha (s_{i'})}(s_{i'})$.
    
    Firstly, observe that as this is the first time $(s_i,c_j)$ is being added to $\mathcal{F}$, $s_i$ is currently matched to $\alpha(s_i)$, and thus if $\alpha(s_i)$ were the rightmost college in the leximin tuple of $\mu$, then it would be preferred over $s_{i'}$ by $\alpha(s_i)$ and wouldn't be added to $\mathcal{F}$. Consequently, $\alpha (s_i)$ is the leftmost college and $c_j=-\alpha(s_i)=\alpha (s_{i'})$. Now, as a result, any stable matching in which $s_i$ is being matched to $-\alpha (s_i)$, so is $s_{i'}$. Thus, from case (ii), any such matching will not have leximin value greater than $\mu_p$.

\end{enumerate}

As the leximin optimal stable matching will have a leximin value at least as much as $\mu_P$, the matching does not match any forbidden pairs from $\mathcal{F}$.
\end{proof}

Now, it suffices to show that for each student, if they were matched differently from what the algorithm returns, the matching would either be unstable or leximin inferior.

\algconst*

\begin{proof}
Let $I$ be an SMO instance. Fix an arbitrary student $s_i$, for some $i\in [n]$ and let $\mu^*$ be the matching returned by FaSt-Const (Algorithm \ref{alg:m2}) and let $\mathcal{F}$ be as at the end of the execution of FaSt-Const on $I$. We consider three possible cases for the matching of $s_i$ and show that in each matching $s_i$ differently would have led to either an unstable matching or one that is leximin inferior.\\

\noindent \textbf{Case 1:} $\mu^*(s_i)=c_j$ and $(s_i,c_{-j})\in \mathcal{F}$\\

\noindent From Lemma \ref{lem:forbid}, we have that any matching that matches $s_i$ to $c_{-j}$ is either unstable or leximin inferior to $\mu^*$.\\

\noindent \textbf{Case 2:} $\mu^*(s_i)=\alpha(s_i)$ and $\{(s_i,\alpha (s_i)),(s_i,-\alpha (s_i))\}\cup \mathcal{F}=\emptyset$. \\

\noindent This implies that $s_i$ was never considered for $-\alpha (s_i)$ during the entire execution of the algorithm. This is possible in one of two ways:
\begin{enumerate}
\item $\alpha(s_i)$ is the leftmost college in the leximin tuple of $\mu^*$. Hence, any stable matching $\mu$ where $s_i$ is not matched to would either match a forbidden pair or give $\alpha (s_i)$ even lower value than in $\mu^*$, making $\mu$ leximin inferior to $\mu^*$ OR

\item $\alpha(s_i)$ is the rightmost college in the leximin tuple of $\mu^*$. Let $s_{i'}=Bottom(\mu^*(\alpha(s_i),\alpha(s_i))$. Thus, we have that $(s_{i'},-\alpha (s_i))\in \mathcal{F}$. Now any stable matching $\mu$ which matches $s_i$ to $-\alpha (s_i)$ also matches $s_{i'}$ to $-\alpha (s_i)$, which are forbidden, hence $\mu$ will be leximin inferior to $\mu^*.$
\end{enumerate}

\noindent\textbf{Case 3}: $\mu^*(s_i)=\alpha (s_i)$, $(s_i,-\alpha(s_i))\notin \mathcal{F}$ but $(s_i, \alpha (s_i))\in \mathcal{F}$.\\

\noindent This implies that $s_i$ was considered for $-\alpha(s_i)$ during the execution of FaSt-Const but the matching was lower in leximin value to $\mu^*$. This implies till the end of the execution of FaSt-Const, the leximin value of the matchings considered were not higher than $\mu^*$, and the execution stopped due to some other forbidden pair. This implies that any stable matching where $s_i$ is matched to $-\alpha(s_i)$ which has not been explored by during the execution of the algorithm requires matching a forbidden pair or sending students from the leftmost college to the rightmost. In either case the matching will be lower in leximin value to $\mu^*$. 
This completes the proof
\end{proof}

\textbf{Time Complexity}: Observe that each iteration of the while loop toggles the matching of one student at a time and no student is ever moved twice. Within each step takes at most time $O(n)$, as a result, FaSt-Const takes time at most $O(n^2)$.

We now show that in more general settings, especially in the absence of strict rankings, the problem of finding the leximin optimal stable matching becomes intractable

\section{Intractability without Strict Rankings}\label{sec:hard}
We shall now establish the necessity of strict preferences and rankings by establishing hardness otherwise. We shall set $b_j=n-m+1$ in these reductions. We first look at a strict preferences setting without rankings. We find a reduction from the Subset Sum problem (SSP) to this, making it NP-Hard. In SSP, given a set of $k$ positive integers, $A=\{a_1,\cdots,a_k\}$, and a target integer value $B$, such that $\max_A a_i < B \leq \sum_A a_i$, we must decide whether there exists a subset $S\subseteq A$, such that $\sum_S a_i=B$.
    
    \stricthard*
    
    \begin{proof}
    
        Given an instance of SSP with $A=\{a_1,\cdots,a_k\}$ and $B$ we construct an instance of stable many-to-one matchings (SMO) as follows: Set the number of \nameBs $m=k+1$ and the number of \nameAs $n=2k$. We shall construct one \nameA $\A_i$ for each integer $a_i$. The most preferred \nameB of $\A_i$ will be $\B_i$. We shall define $\B_m$ to be the \nameB whose matching (set of students) will correspond to the subset selected for SSP. The remaining \nameAs ensure that all the $\nameBs$ $\B_j$ which are not matched to their corresponding $\A_j$ are not unmatched. These students will be $s_{k+1},\cdots, s_{2k}$. For each $j\in  [k]$, the most preferred college of $s_{j+k}$ will be $c_j$. For stability, the most preferred \nameA of $\B_j$, for $j\in [k]$ is $\A_{j+k}$, the second most preferred is $\A_j$. Consequently, with $\epsilon = \frac{1}{3k^2} < \frac{1}{n}$, the valuations are defined as follows:
        
        For \nameA $\A_i$, $i\in [k]$  if $i=j$, set $u_i(\B_j)=B$. If $j=m$, set $u_i(\B_j)=B-a_i + \epsilon $, else set $u_i(\B_j)=j\epsilon$. For \nameA $\A_i$, $i\in [2k]\setminus [k]$, if $i=j+k$, set $u_i(\B_j)=B$, else, $u_i(\B_j)=j\epsilon$. 
        For \nameB $\B_j$, with $j\in [k]$, if $i=j+k$, set $v_j(\A_i)=2B$, if $i=j$, define $v_j(\A_i)=B$ else, $v_j(\A_i)=i\epsilon$. Finally, for the \nameB corresponding to the subset selected, if $i\in [k]$, define $v_m(\A_i)=a_i$, else, $v_m(\A_i)=i\epsilon$.
    
        Observe that $j\epsilon<1<B-a_i+\epsilon<B$ for all $i\in [k],\,j\in [n]$ and hence, all preferences are strict. As a result, for each $s_i$ for $i\in [k]$, the most preferred college is $c_i$ and the second most preferred college is $c_m$. Further, in any complete \textit{stable} matching $\mu$, we have that for $i\in [2k] \setminus [k]$,  $\mu(\A_i)=\B_{i-k}$. If this doesn't hold for some $i\in [2k]\setminus [k]$,  then $\A_i - \B_{i-k}$ is a blocking pair. 
        
        For $i\in [k]$, for any complete stable matching $\mu$, if  $\mu(\A_i) \notin \{\B_i,\B_m\}$, $\mu$ is leximin dominated by the matchings which are identical to $\mu$ except $\A_i$ is matched to one of $\B_i$ or $\B_m$. Thus, we shall henceforth only consider stable matchings $\mu$ where for $i\in [k]$, $\mu(\A_i) \in \{\B_i,\B_m\}$ and for $i\in [2k] \setminus [k]$,  $\mu(\A_i)=\B_{i-k}$.
        Let $\mu^{i+}$ denote the matching that is identical to $\mu$ except that $\A_i$ is matched to $\B_i$, i.e. if $i'\neq i$, $\mu^{i+}(\A_{i'})=\mu(\A_{i'})$, and $\mu^{i+}(\A_i)=\B_i$. Similarly define $\mu^{i-}$ s.t. if $i'\neq i$, $\mu^{i-}(\A_{i'})=\mu(\A_{i'})$, and $\mu^{i-}(\A_i)=\B_m$. Observe that if $\mu$ is stable, so are $\mu^{i+}$ and $\mu^{i-}$. 
        
        We shall now show that the leximin optimal stable matching will try to ensure that $\B_m$ gets value $B$. To this end, we consider arbitrary, complete stable matchings and show when would moving one \nameA, say $\A_i$ either to $\B_m$ or $\B_i$ improve the leximin value. To show that one matching leximin dominates another, here, we need only consider the agents whose values change: $\A_i$, $\B_i$ and $\B_m$, and show that the minimum of these agents' values is more in one matching than the minimum of these agents' values in the other.
        
        Now fix a complete stable matching $\mu$. If $\B_m$'s value for $\mu$ is greater than $B$, that is $v_m(\mu)> B$, then for any student $\A_i\in \mu(\B_m)$, they lie to the left of $\B_m$ in the leximin tuple with value $u_i(\mu) = B - a_i + \epsilon < B < v_m(\mu)$. Now, it is easy to see that $v_m(\mu^{i+})=v_m(\mu)-a_i\geq B+1-a_i>u_i(\mu)$. The valuations of all other agents remain unchanged or increase. As a result, if $v_m(\mu)> B$, then for any $\A_i\in \mu(\B_m)$, $\mathcal{L}_{\mu^{i+}}>\mathcal{L}_{\mu}$.
        
        On the other hand if $v_m(\mu)\leq B$, then for any $\A_i\in \mu(\B_m)$, $u_i(\mu)=B-a_i + \epsilon>B-a_i\geq v_m(\mu) -a_i=v_m(\mu^{i+})$. Thus, $\mu$ is leximin superior to $\mu^{i+}$ for any $s_i\in \mu(B_m)$. For any $\A_i\notin \mu(\B_m)$ such that $a_i\leq B-v_m(\mu)$, we find that $u_i(\mu^{i-})=B-a_i+\epsilon>B-a_i\geq v_m(\mu)$. As a result, $\mathcal{L}_{\mu^{i-}}>\mathcal{L}_{\mu}$ for any $\A_i\notin \mu(\B_m)$ such that $a_i\leq B-v_m(\mu)$. Now for $\A_i$ such that $a_i> B-v_m(\mu)$, $u_i(\mu^{i-})=B-a_i+\epsilon \leq B - (B-v_m (\mu) +1) +\epsilon < v_m(\mu)$. Thus, $\mathcal{L}_{\mu^{i-}}<\mathcal{L}_{\mu}$.
        
        Consequently, the leximin optimal stable matching $\mu^* $ will always maximize $v_m(\mu)$ with the constraint that $v_m(\mu^*)\leq B$ and will match $\B_m$ to the set of students corresponding to the subset $S^* =\argmax_{S\subseteq A, \sum_S a_i\leq B} \sum_S a_i$. Thus, the required subset exists if and only if $v_m(\mu^*)=B$.    
    \end{proof}
    
As a result, for unrestricted $m$, with strict preferences, without rankings the problem of finding a leximin optimal remains intractable. In fact, the preferences of the agents in the instances created in the above proof are actually very similar with a few changes only.  We now find that under isometric valuations and weak rankings (strict preferences need not be satisfied) , finding the leximin optimal stable matching is NP-Hard, even with a constant number of colleges.

    


    \isohard*
    
   \begin{proof} 
 
For the $m=2$ setting, we give a reduction from the balanced partition problem, which is known to be NP-Complete.\\

\noindent\textbf{Balanced Partition Problem:} Given a set of integers $P = \{p_1, \cdots,$  $ p_k\}$, such that $\sum _{i=1}^k p_i=2B$,  find a 2-partition $\{A_1, A_2\}$ of $P$ which satisfies $\sum _{p_i\in A_j} p_i=B$ for all $j=1,\,2$. \\

\noindent  Given $P$, we shall construct an instance of isometric valuations $I$ with $m=2$, such that $P$ admits a balanced partition if and only if the leximin optimal stable matching of the instance $I$ allocates valuation of $B$ to both $\B_1$ and $\B_2$. Set $n=k$,  $V_{ij}=p_i$ for all $i\in [n], j=1,\,2$. In such a setting it is easy to see that all complete matchings are stable as no \nameA has any incentive to deviate.
In any matching, \nameAs will always get the same value, that is \nameA $\A_i$ always gets value $p_i$. Thus for a leximin optimal matching, it suffices to check the values that the \nameBs attain.  

Let $P$ admit a balanced partition $\{A_1,A_2\}$. Let $\mu$ be the matching which matches $\B_j$ to all the \nameAs whose values are in $A_j$ for all $j=1,\, 2$. Here, clearly, all \nameBs get value $B$. The leximin tuple of $\mu$ will list the $p_i$ values first, in non-decreasing order and the last two entries will all be $B$. Any matching that gives any one of the colleges, say $\B_1$, higher valuation, will naturally decrease the valuation of the another college. Hence $\B_2$'s value will either be lower in the same position in the leximin tuple, or be to the left, resulting in a lower leximin value in both cases. Hence, $\mu$ is a leximin optimal matching.
    
Conversely, if the leximin optimal matching gives value $B$ to all \nameBs then the partition created by the matching is clearly the required balanced partition.

For the $n=3m$ case we give a reduction from the 3-Partition problem which is known to be strongly NP-Hard.\\ 

\noindent\textbf{3-Partition Problem:} Given a set of integers $P = \{p_1,\ldots,\, p_k\}$, such that $\sum _{i=1}^k p_i=\sfrac{3t}{k}$,  find a $\sfrac{k}{3}$-partition $\{A_1, \cdots, A_{\sfrac{k}{3}}\}$ of $P$ which satisfies $\sum _{p_i\in A_j} p_i=t$ for all $i\in [\sfrac{k}{3}]$. \\

Given an instance of the 3-partition problem $P$, we create a matchings instance as follows. Set $m=\sfrac{k}{3}$, $n=k$ and set $V_{ij}=p_i$ for all $j\in [m]$. 
Observe that this is identical to the previous construction, with the exception that there are now more colleges. Thus, from the same reasoning, for a leximin optimal matching, it suffices to check the values that the \nameBs attain.
    
Let $P$ admit a 3-partition $\{A_1,A_2,\cdots, A_{\sfrac{n}{3}}\}$. Let $\mu$ be the matching which matches $\B_j$ to all the \nameAs whose values are in $A_j$ for all $j \in [m]$. Here, clearly, all \nameBs get value $t$. The leximin tuple of $\mu$ will list the $p_i$ values first, in non-decreasing order and the last $m$ entries will all be $t$. Any matching that gives any one of the colleges, say $\B_j$, higher valuation, will naturally decrease the valuation of the another college, say $\B_{j'}$. Hence $\B_{j'}$'s value will either be lower in the same position in the leximin tuple, or be to the left, resulting in a lower leximin value in both cases. Hence, $\mu$ is a leximin optimal matching.
    
Conversely, if the leximin optimal matching gives value $t$ to all \nameBs then the partition created by the matching is clearly the required 3-partition.
\end{proof}

\subsection*{Hardness of Approximation}
We now discuss the possibility of finding some approximation to the leximin optimal stable matching. 
We shall show that under weak rankings, unless P=NP, no polynomial factor approximation is possible. We first define what an approximation algorithm for finding a leximin optimal (over any space) must guarantee.

\begin{definition}[$\alpha$-approximation of leximin] An algorithm is to give an $\alpha$-approximation $\alpha \in (0,1)$ to the leximin optimal, if given instance $I$ with $\mu^*$ as the leximin optimal solution, it outputs $\mu$ s.t. for each index $t$ of the leximin tuple, $\alpha \mathcal{L}_{\mu^*}[t] \leq \mathcal{L}_{\mu}[t] \leq \frac{1}{\alpha} \mathcal{L}_{\mu^*}[t]$ 

\end{definition}

\noindent This is a generalization of the definition of approximation algorithms which optimize for a single value. Note that it important that the approximation factor holds for every index. We would not like a setting where the first index of the leximin tuple satisfies the condition but the subsequent indices do not. This is because a good approximation to the first index may be satisfied by all stable matchings and would not guarantee any fairness to any of the other agents. 
We now give a reduction from the Bin Packing problem to establish the hardness of approximation for leximin optimality. We look at the decision version of the problem where given a set of items $G$, each with weight $w_i$ and $k$ bins of capacity 1, we must decide if there is a partition of the items into $k$ bins such that the sum of the weights of the items in each bin is at most 1. We assume, without loss of generality, that $w_i\leq 1$ and $|G|\geq k >1$.

Before proving Theorem \ref{thm:apxhard}, we will develop some of the necessary machinery for its proof in the following lemma.

\begin{lemma}\label{lem:halfapx}
Unless P=NP, no $1/2+\epsilon$ approximation algorithm exists for finding the leximin optimal stable matching
\end{lemma}

\begin{proof}
Given a bin packing instance $\langle G,k, \{w_i\}_{g_i\in G}\rangle$, where $0 \le w_i\leq 1$ for all items $g_i\in G$, $\sum_{g_i\in G} w_i\leq k$, and  $|G|>k$, we construct an instance of stable many-to-one matching (SMO) as follows:  Set $n=|G|+1$, we create a \nameA for each item and an additional dummy student. Further, set $m=k+1$, we create a \nameB for each bin and an additional dummy \nameB which all the \nameAs will prefer to the bins. 
We shall interchangeably refer to these non-dummy \nameAs and \nameBs as items and bins. 

The valuations are then set so that the following properties hold:

\begin{enumerate}
    \item[P1.] All the \nameAs $s_1, \cdots, s_{n-1}$ prefer the dummy \nameB $c_n$ over other colleges,
    \item[P2.] All the non-dummy \nameBs (bins) $c_1,\cdots, c_{m-1}$ have no value for the dummy student $s_n$,
    \item[P3.] The dummy \nameA $s_n$ and dummy \nameB $c_m$ are always matched in a leximin optimal stable matching, and
    \item[P4.] At least one of the non-dummy \nameAs (items) $s_1, \cdots, s_{n-1}$ will be matched to the dummy \nameB $c_m$ if and only if there is no bin packing.
\end{enumerate}  

We set the valuations as follows: 
Each non-dummy \nameA or item's valuation towards the non-dummy \nameBs or bins are equal and given by $u_i(c_j)=1-w_i+\epsilon, \mbox{ for each } i\in [n-1] \mbox{ and each } j\in [m-1]$, where, $\epsilon>0$ is an infinitesimally small value. 
Moreover, to ensure that the \nameAs prefer the dummy \nameB $\B_{m}$ over the bins, we set $u_i(\B_m)=2$ for each $i\in [n-1]$. For the dummy student $s_n$, we set $u_n(\B_j)=0$ for all $j\in[m-1]$ and $u_n(\B_m)=1$. 

The bins' valuations for each item $i$ is the weight $w_i$, i.e, $v_j(\A_i)=w_i$ for each $j\in [m-1]$ and each $i\in [n-1] $. Also, for any bin, the value for the dummy \nameA is $v_j(\A_n)=0$. 
For the dummy \nameB $\B_m$, we set $v_m(s_i)=n$ for all $i\in [n]$. If $\A_n$ is not assigned to $\B_n$, the valuation obtained by $\A_n$ would be $0$. Hence, the lexicographically ordered valuations obtained by assigning $\A_n$ to $\B_m$ would always dominate all such matchings where all other agents are the same but $\A_n$ is assigned to a bin. Hence, a leximin optimal stable matching will always match $\A_n$ to $\B_m$.

We now show that, in the leximin optimal stable matching, each item is assigned to some bin if and only if a bin packing exists. 
Suppose the bin packing instance $\langle G,k, \{w_i\}_{g_i\in G}\rangle$ admits a bin packing. We claim that the leximin optimal stable matching, matches all the items to bins. Let $\mu$ be the leximin optimal over the space of all stable matchings  where the items are matched to bins only. Note that $\mu$ is also the solution obtained by taking leximin optimal of bin valuations over all bin packing solutions. Thus, for each bin $\B_j$, we have $v_j(\mu)\leq 1$. 

Now, fix an item $s_i$, $i\in [n-1]$. As $s_i$ is matched to a bin, say $c_j$, under $\mu$, $u_i(\mu)=1-w_i+\epsilon$. Now let $\mu'$ be identical to $\mu$, with the exception that $s_i$ is matched to $\B_m$. As a result, the value of $c_j$ for $\mu'$ is $v_j(\mu')=v_j(\mu)-w_i\leq 1-w_i<u_i(\mu)$ and valuation of $s_i$ would increase to $2$, and $v_m(\mu')=2n$. The values of all other agents are identical in $\mu$ and $\mu'$. Thus $\mu'$ is leximin inferior to $\mu$. To overcome this, we must take an item $\A_{i'}$ from another bin, which decreases its valuation, in turn to less than $1-w_{i'}+\epsilon$, thus this matching is also leximin inferior to $\mu$. Hence, every time we attempt to cover for the decrease in the leximin value, will create a new bin whose valuation is less than the valuation of the item removed from it. This will always be leximin inferior to $\mu$. Consequently, if a bin packing exists, the leximin optimal matching must match the items to the bins only, and hence, $\B_m$'s valuation will be $n$.

Now suppose that no bin packing solution exists, any matching that matches all the items to bins only, will give at least one bin valuation strictly greater than 1. Consider any item matched to this bin. Matching it to $\B_m$ will reduce the bin's valuation but it will still remain greater than $1-w_i +\epsilon$, $\A_i$'s current valuation ($\epsilon$ is infinitesimally small). Along with this, the valuation of $\A_i$ and $\B_m$ will both increase and the remaining agents' will not be affected. As a result, whenever a bin packing does not exist, at least one item will be matched to $\B_m$, and thus $\B_m$'s valuation will be at least $2n$. This shows that the valuations ensure that properties P1-P4 are true. 

Suppose, for some $\alpha>1/2$, an $\alpha$-approximation algorithm, say $ALG$, exists for finding the leximin optimal stable matching. Now, let $\mu'$ be the matching that $ALG$ outputs on the matching instance $I$ corresponding to the bin packing problem and let $\mu$ be the leximin optimal stable matching. Recall that by construction of $I$, $\B_m$'s valuation always appears highest (last) in the leximin tuple.  Therefore, whenever a bin packing exists $v_m(\mu)=n$ and thus, $\alpha n<v_{m}(\mu')<\frac{n}{\alpha}$. Since $\alpha>1/2$, $\frac{n}{2}<v_{m}(\mu')<2n$. Moreover, if a bin packing does not exist, $v_m(\mu)\geq 2n$, and thus, $v_m(\mu')>n$. Further, by construction, $v_m(\cdot)$ can only have values which are integral multiples of $n$, thus $v_m(\mu')\geq 2n$ whenever a bin packing does not exist.

Hence, unless P=NP, no $\frac{1}{2}+\epsilon$-approximation algorithm exists, for any $\epsilon>0$.
\end{proof}

\noindent We can extend this technique by doing a much more intricate reduction to establish the following theorem.    
    
\polyapxhard*

    \begin{proof}
This proof builds upon the proof of Lemma \ref{lem:halfapx}. In order to get a better bound for the hardness of approximation, we must increase the gap in $\B_m$'s valuations while ensuring that $\B_m$'s valuation always appears at the end of any complete\footnote{Recall that we say a matching is complete if each agent's matching is non-empty} stable matching. We do this by replicating the bin packing instance. 

Thus, for any $\delta>0$, $c\in \mathbb{Z}^+$, we can set $t=\lceil cn^{\delta }\rceil$. Now, given a bin packing instance $\langle G,k, \{w_i|g_i\in G\}\rangle$, we create an SMO instance as follows. Let $|G|=l$.

Set $n=tl+1$ and $m=tk+1$ and hence, replicate the bin packing instance $t$ times. That is, for each item we create $t$ copies of it and now there are $tk$ bins. We shall keep these copies from interacting with each other, that is, items from one copy will only be matched to bins from the same copy. The valuations will continue to satisfy the properties P1-P4 listed in the proof of the previous lemma. The valuation functions are accordingly defined as follows:

For the $p^{\text{th}}$ copy of bin $j$, $c_{(p-1)k+j}$, it has value for items from the same or any earlier copy, hence we set $v_{(p-1)k+j}(\A_{(p'-1)l+i})=w_i$ for all $j\in[k],\,p,\,p'\,\in [t],\, i\in [l] $  s.t. $ p\geq p'$ and $v_{(p-1)k+j}(\A_{(p'-1)l+i})=0$ for all $j\in[k],\,p,\,p'\,\in[t],\, i\in[l] $  s.t. $ p< p'$. None of the bins have value for the dummy student $s_n$, so we set $v_j(s_n)=0$ for all $j\in [m-1]$. For the dummy college $c_m$, all students give equal value and we set $v_m(s_i)=(t+1)n$ for all $i\in [n]$.

For the $p^{\text{th}}$ copy of item $i$, $s_{(p-1)l_i}$, it has no value for bins from later copies and equal value for all bins from the same or earlier copies. To this end, we set $u_{(p-1)l+i} (\B_{(p'-1)k+j})=1-w_i+\epsilon $ for all $ i\in [l], p, p' \in [t], j \in [k] $ s.t. $ p \geq p'$ and $u_{(p-1)l+i} (\B_{(p'-1)k+j})=0 $ for all $ i\in [l], p, p' \in [t], j \in [k] $ s.t. $ p < p'$. All items have the highest value for the dummy college $c_m$, so we set $u_i(c_m)=t+1$ for all $i\in [n-1]$. The dummy student $s_n$ has value only for the dummy college $c_m$, consequently, we define $u_n(c_j)=0$ for all $j\in [m-1]$ and $u_n(c_m)=1$. 

These are summarised as follows.

{
\begin{align}
    v_{(p-1)k+j}(\A_{(p'-1)l+i})&=w_i &&\text{for all }j\in[k],\,p,\,p'\,\in[t],\, i\in[l] \text{ s.t. } p\geq p' \label{eq:binutil}\\
    v_{(p-1)k+j}(\A_{(p'-1)l+i})&=0 &&\text{for all }j\in[k],\,p,p'\in[t],\, i\in[l]\text{ s.t. } p<p' \label{eq:cross1} \\
    v_{j}(\A_n)&=0 &&\text{for all }j\in [m-1] \nonumber  \\
    v_m(\A_i)&=(t+1)n &&\text{for all }i\in [n] \nonumber \\
    u_{(p-1)l+i} (\B_{(p'-1)k+j})&=1-w_i+\epsilon && \text{for all } i\in [l], p, p' \in [t], j \in [k] \text{ s.t. } p \geq p' \label{eq:itemuputil} \\
    u_{(p-1)l+i} (\B_{(p'-1)k+j})&=0 && \text{for all }i\in [l],p,p'\in [t],j\in [k]\text{ s.t. } p < p' \label{eq:itemlowutil}\\
    u_{i}(\B_m)&=t+1 &&\text{for all }i\in [n-1] \nonumber \\
    u_n(\B_j)&=0 && \text{for all }j\in [m-1] \nonumber \\
    u_n(\B_m)&=1 \nonumber
\end{align}
}
\noindent The construction is similar to the one in the proof of Lemma \ref{lem:halfapx}. Here, we make $t$ copies of each item and bin so that we get $t$ copies of the bin packing instance in the SMO instance, with for each $p\in[t]$,  $\A_{(p-1)l+1},\cdots,\A_{(p-1)l+l}$ representing $g_1,\cdots,g_n$ respectively and $\B_{(p-1)k+1},\cdots, \B_{(p-1)k+k}$ representing the $k$ bins of capacity 1. For each $p\in [t]$, we consider $\A_{(p-1)l+1},\cdots,\A_{(p-1)l+l}$ and $\B_{(p-1)k+1},\cdots, \B_{(p-1)k+k}$ to represent the $p^{th}$ copy. 

The valuations are defined such that all agents in a copy $p$ have value 0 for all agents in copy $p'$ for all $p'>p$ (from Equations \ref{eq:cross1} and \ref{eq:itemlowutil}). Consequently, no leximin optimal matching will match a \nameA in copy $p$ to a \nameB in copy $p'$ when $p'\neq p$, as this would mean this \nameA gets value 0. Further, from Equation \ref{eq:itemuputil}, for all $p\in[t]$, each \nameA  in copy $p$, is indifferent between all the \nameBs in copies $1,\cdots,p$. However, the \nameBs in copies $1,\cdots,p-1$ have valuation $0$ for \nameAs in copy $p$ (from Equation \ref{eq:cross1}), ensuring that all items are matched to bins in their own copy or $\B_m$.

Thus, from an analogous argument to that in the proof of the previous lemma, if a bin packing exists, the items are matched to bins (of the same copy) in the leximin optimal stable matching. Thus, it gives $\B_m$ is matched only to $\A_n$ and thus has a valuation of exactly $(t+1)n$. If a bin packing doesn't exist,  at least one item from each copy is matched to $\B_m$ giving it a valuation of at least $(t+1)^2 n$.

We shall now show that a $1/cn^{\delta}$-approximation algorithm with $\delta>0$ and $c\in \mathbb{Z}^+$ to find a leximin optimal stable matching with weak rankings will match $c_m$ to at least one item if and only if a bin packing doesn't exist. Recall that we set $t=\lceil cn^{\delta }\rceil$. Now let a $1/cn^{\delta }$-approximation algorithm exist, call it $ALG$. 

Further, let $\mu^*$ denote a leximin optimal stable matching of the instance constructed. Then clearly $u_i(\mu^*)>0$ and $v_j(\mu^*)>0$ for all $i\in [n]$ and $j\in [m]$. Further note that by construction, $\B_m$'s valuation will always appear in the last index of the leximin tuple of any matching where $\B_m$ is matched with at least one \nameA. As $1/cn^{\delta}>0$, $v_m(\mu_{ALG})\geq (t+1)n$ and it will be the last index in the leximin tuple of $\mu_{ALG}$. Similarly for the dummy student $s_n$, $u_n(\mu_{ALG})$>0. Consequently, $\mu_{ALG}(s_n)=c_m$. Further, there do not exist $p,p'\in [t]$, $i\in [l]$ and $j\in[k]$ s.t. $p\neq p'$ and $\mu_{ALG}(\A_{(p-1)l+i})=\B_{(p'-1)k+j}$.

Now let a bin packing exist. Hence, for all students, that is for all $i\in [n]$,  $0<u_i(\mu^*) \leq 1$ and for all bins $j\in[m-1]$ $0<v_j(\mu^*)\leq 1$. Thus, for all $d\in[m+n-1]$, $\mathcal{L}_{\mu^*}[d]\leq 1$. Further, $v_m(\mu^*)=(t+1)n$. 

Now as $\B_m$'s valuation will be the greatest in $\mu_{ALG}$, for any \nameA $\A_i,\,i\in [n-1]$, $u_i(\mu_{ALG})$ must appear in the first $m+n-1$ entries of $\mathcal{L}_{\mu_{ALG}}$, say $d_i$.  For $i\in [n-1]$, let $d_i$ be index where $s_i$'s value lies in the leximin tuple for $\mu_{ALG}$. Thus, we have that  $u_i(\mu_{ALG}) =\mathcal{L}_{\mu_{ALG}}[d_i]\leq cn^{\delta}\mathcal{L}_{\mu^*}[d_i]\leq cn^{\delta}<t+1$

Thus, whenever there is a bin packing, no $\A_i$, for $i\in [n-1]$, is matched to $\B_m$. Consequently, the dummy college gets value $v_m(\mu_{ALG})=(t+1)n$.

If a bin packing does not exist, we have that $c_m$ is matched to at least one item from each of the $t$ copies under $\mu^*$ and $v_m(\mu^*)\geq t(t+1) n$. Further, as $\B_m$ is the last index of the leximin tuple in both $\mu^*$ and $\mu_{ALG}$, we have that, $v_m(\mu_{ALG})\geq (1/cn^{\delta})t(t+1) n$. Now as $t=\lceil cn^{\delta}\rceil$, we have that \[v_m(\mu_{ALG})\geq \frac{t(t+1) n}{t}>(t+1)n.\] 
Thus, a bin packing exists if and only if $v_m(\mu_{ALG})=(t+1)n$.
\end{proof}

Note that the above reduction constructs a SMO instance with weak rankings. Thus, even when the instances have weak rankings, the hardness of approximation remains. This clearly subsumes the case when there are no consistent rankings. This concludes all the results of this paper.

\section{Incentive Compatibility} \label{sec:ic}

We now explore the existence of incentive compatible mechanisms which return the leximin optimal over the space of stable matchings. Given our results on the tractability of this problem, we restrict our focus to instances with raked valuations. Observe that this an impossibility here also implies an impossibility when the valuations are unrestricted.

Narang and Narahari\cite{narang2020study} study the existence of incentive compatible mechanisms which find stable fractional matchings when $m=n$. In particular, they study a class of matchings instances where incentive compatibility is achievable. Interestingly the space of instances with rankings is subsumed by this class for one-to-one matchings. We show that for any mechanism that computes the leximin optimal stable matching, agents have an incentive to be honest if and only if $n=m$, in which case, we are essentially in the setting studied by Narang and Narahari\cite{narang2020study}.

\begin{restatable}{lemma}{icimp}\label{lem:ic}
A mechanism which takes as input an SMO instance $I$ with general ranked valuations and outputs the leximin optimal stable matching does not give any agent an incentive to misreport their valuations if and only if $n=m$.
\end{restatable}

\begin{proof}

We first begin with the simpler case. Under strict ranked valuations, if $n=m$ then there is a unique stable matching, irrespective of the exact valuations reported. Thus, no agent has an incentive to misreport their preferences.

Now we look at the case when $n>m$ with rankings. We shall show that there always exists an agent who wishes to misreport their valuations.

Given $I=\langle \symbA, \symbB, U, V \rangle$, let $Best:\symbA \rightarrow  \symbB$ be a function such that $Best(\A_i)$ is the highest ranked \nameB,  $\A_i$ can be matched to in a stable matching where each agent's matching is non-empty. Similarly let $Worst$ be a a function such that $Worst(\B_j)$ is the lowest ranked \nameA, $\B_j$ can be matched to in a stable matching where each agent's matching is non-empty.

Therefore $Best(\B_j)=\A_j$ for all $j\in [m]$. For $i\leq n-m+1$, $Best(\A_i)\B_1$ and for $i\geq n-m+1$, $Best(\A_i)=\B_{m-n+i}$. Thus if $i\neq 1$, $Best(s_i)\neq Best^{-1}(s_i)$.

Let $\mu^*$ be the leximin optimal stable matching for $I$. Fix $p$ such that $1<p\leq m<n$ \\

\noindent \textbf{Case 1:} $\mu^*(\A_p)=Best(\A_p)$. 

Let $\B_t=Best^{-1}(\A_p)$. Define $\gamma=\min \{u_i(\B_j)|\,u_i(\B_j)>0\}\cup \{v_j(\A_i)|\,v_j(\A_i)>0\}$

Set $\alpha = \Big\lceil \frac{v_t(\A_p)}{\gamma}\Big\rceil n$. Define $v'$ where \[v'(\A_i)=\begin{cases}
v_t(\A_i) & i>p \\
v_t(\A_i)/\alpha & i\leq p

\end{cases}
\]
Now $\B_t$ can misreport their valuation function as $v'$ and ensure they are matched to $\{\A_p,\cdots, Worst(\B_j)\}$.  This is because our choice of $v'$ ensures that $\B_t$ will be considered as the agent with the lowest valuation. Thus $\B_t$ has an incentive to misreport\\

\noindent \textbf{Case 2:} $\mu^*(\A_p)\neq Best(\A_p)$.

Let $\delta =\min_{j>1} v_j(Worst(\B_j))$. Recall that if $u_p(Best(\A_p))<\delta$ then in the leximin optimal stable matching, $\A_p$ must be matched to $Best(\A_p)$. Thus, $\A_p$ can misreport $u_p$ as $u'$ where $u'(\B_j)=\frac{\delta -\epsilon}{j}$, $\epsilon>0$.

This ensures that $\A_p$ is matched to $Best(\A_p)$ by misreporting. Thus, $\A_p$ has an incentive to misreport.

\end{proof}
\section{Conclusion and Future Work} 


This paper initiates the study of finding a fair and stable many-to-one  matching under cardinal valuations. We studied leximin optimality, which has been hitherto unexplored for stable matchings. We give efficient algorithms to find the leximin optimal stable matching under rankings with strict preferences. In the absence of rankings, but with strict preferences, we give an algorithm to find the leximin optimal stable matching when the number of colleges is fixed to two.  We then showed that when $n=\Omega (m)$, then, under strict preferences without rankings, the problem becomes intractable.
We showed that even with weak rankings, finding a leximin optimal stable matching is intractable. In fact, it is NP-Hard to find even a polynomial approximation with general valuations. 

One open problem is whether for a higher constant number of \nameBs or even a sublinear number of colleges, with strict preferences, an exact algorithm may be possible. Another potential direction may be finding an a more general subclass of matching instances where it is possible to find a fair or approximately-fair and stable matching. While our results have ruled out approximations for additive valuations in general, under isometric valuations or general strict preferences, approximations may still be possible. One line of future work would be to explore this direction. We believe that our work will encourage further research on fair and stable matchings. 

\newpage

\bibliographystyle{abbrv} 
\bibliography{references.bib}
\newpage

\appendix
\section{Other Fairness Notions} \label{sec:envy}
\textbf{Envy and Stability}: Fairness has been widely studied in computational social choice, with various fairness notions considered for both divisible and indivisible goods. Since we are interested in integral matchings (matchings where agents are matched wholly/ integrally, no agent is matched partially/fractionally), we shall only discuss allocations of indivisible goods. 
Often, the first fairness notion that comes to mind when we think of fair allocations is \textit{envy-freeness} (EF). Informally, an EF allocation guarantees that every agent would prefer its own allocation over any other agent's allocation.  
Unfortunately, EF allocations/matchings need not exist in indivisible settings. 
Consider $\symbA=\{\A_1,\A_2\}$ and $\symbB=\{\B_1,\B_2\}$ with isometric valuations $V_{11}=10$, $V_{12}=5$, $V_{21}=8$, and $V_{22}=2$. Let $\mu$ be a matching such that $\mu(\A_1)=\B_1$ and $\mu(\A_2)=\B_2$. Since $u_{2}(\B_1)>u_{2}(\B_2)$, the agent $\A_2$ envies $\A_1$, similarly, agent $\B_2$ envies $\B_1$. It is easy to see that in every possible matching $\mu$, there will always be at least one agent who will envy the other agent.  

Envy-freeness upto one item (EF1) is a popular notion of fairness for allocations of indivisible items. An allocation $A=(A_1,\cdots,A_n)$ is said to be EF1 if for every distinct pair of agents $i$ and $j$, there exists $g\in A_j$ such that $u_i(A_i)\geq u_i(A_j\backslash \{g\})$. An allocation is \textit{envy-free up to any good} (EFX) if for all $g\in A_j$, $u_i(A_i)\geq u_i(A_j\backslash \{g\})$. In one-to-one matchings, EF1 and EFX are achieved trivially. 

In the case of many-to-one matchings, we can find matchings that are EF1 for the \nameBs  by using a round robin procedure. However, such a matching need not be stable, even under ranked isometric valuations. In fact, there exist ranked isometric valuations instances where no matching simultaneously satisfies stability and EF1. Consider the following example: let $n=4$ and $m=2$. The valuation matrix is as in Table \ref{tab:ex}.
{\footnotesize  
\begin{table}[ht]
    \begin{minipage}{0.45\textwidth}
         \centering
         \begin{tabular}{|c|c|c|}
        \hline
        \nameA & $\B_1$ & $\B_2$ \\
        \hline
        $\A_1$ & 100 & 10 \\
        $\A_2$ & 99 & 9 \\
        $\A_3$ & 20 & 4 \\
        $\A_4$ & 19 & 3 \\
        \hline
        
    \end{tabular}
    \caption{Valuation Matrix}
    \label{tab:ex}
    \end{minipage}
    \begin{minipage}{0.45\textwidth}
    \centering
    \begin{tabular}{|c|c|c|c|c|}
    \hline
    $\B_1$ & $\B_2$ & $E_\symbA$ & $E_\symbB$ & $E_{total}$  \\
    \hline
    1-4 & - & 0 & 26 & 26 \\
    1-3 & 4 & 16 & 20 & 36 \\
    1,2 & 3,4 & 32 & 12 & 44\\
    1 & 2-4 & 120 & 38 & 158\\
    - & 1-4 & 0 & 238 & 238 \\
    \hline 
\end{tabular}     
\caption{Stable Matching Space}    
\label{tab:val}
    \end{minipage}
 \   
 \caption*{A counterexample for Envy and Stability}  
    
 \end{table}   
}
\noindent Now it is easy to verify that the only EF1 matching is $\mu=\{(\A_1,\B_1),(\A_2,\B_2),(\A_3,\B_1), (\A_4,\B_2)\}$. This however is not stable as $(\A_2,\B_1)$ form a blocking pair. Hence, when stability is non-negotiable, EF1 cannot be the fairness notion of choice for isometric valuations. Consequently, neither can EFX.

It has been shown that the space of stable matchings of any given instance can be captured as the extreme points of a linear polytope. Thus, we can optimize any linear function over this space. Consequently, our next idea may be to look for a stable matching that minimizes average envy, or equivalently, total envy. But that too can often lead to matchings that are inherently unfair. Consider the example given in Table \ref{tab:ex}. 
The stable matchings in this example and the envy they induce is as in Table \ref{tab:val}.
 
Clearly, in this example, matching all the \nameAs to $\B_1$ reduces the total/average envy but this is obviously unfair to $\B_2$. This is happening despite the fact that there are more \nameAs than there are colleges. Note that in this example, there is in fact a ranking and yet envy doesn't work well. In fact, this example can be extended for much larger values of $n$ so that $\B_2$ is matched to no \nameAs in the stable matching which minimizes total envy.  

\textbf{Welfare based Fairness}: Maximizing for Nash Social Welfare over the space of stable matchings for such examples would again result in matchings where $\B_2$ is matched to exactly 1 \nameA even for very high values of $n$.  These solutions in settings like labour markets or college admissions would result in the rich getting richer, defeating the purpose of fairness. One alternative would be egalitarian welfare, where we aim to simply maximize the valuation of the worst off agent. Clearly, all leximin optimal solutions would also optimize for egalitarian welfare. However, in the case of ranked isometric valuations, as a result of Lemma \ref{lem:struct}, any fair and stable matching must match $\A_n$ to $\B_m$, and $\A_n$ will always have least valuation. As a result, all complete stable matchings under isometric valuations will optimize egalitarian welfare. Thus, in this setting, egalitarian welfare alone is not enough to ensure true fairness. In a similar spirit, a good approximation to egalitarian welfare may be satisfied by matchings which give no guarantee to the remaining agents, and can even be inefficient to a large extent. 
In order to avoid such outcomes, we study leximin optimal fairness.

\section{Capacity Constrained Settings} \label{sec:cap}
In FaSt and FaSt-Gen we had assumed that each \nameB has capacity $b_j=n-1$.  We now give the formal algorithms for when this assumption is relaxed to allow \nameBs to have arbitrary capacity. The Demote procedure (Algorithm \ref{alg:demote}) remains the same as the uncapacitated setting.

\subsection{Ranked Isometric Valuations}
We first translate FaSt for capacity constrained settings. Recall that $b_j\in [n]$ denotes the capacity of $\B_j$. We assume without loss of generality that $b_j\geq 1$. Here, the student optimal stable matching, matches $\B_1$ to as many students as possible, i.e. $\min (n-m+1,b_1)$. If $b_1< m-n+1$, then we match as many students to $\B_2$ as possible i.e. $\min (n-b_1-m+2,b_2)$. This process is now repeated till all the students are matched. Now we can simply follow FaSt as is, taking students from the lowest ranked college with multiple students matched to it,  with the additional constraint that we fix a \nameB when it reaches full capacity.   

In the tie-breaking routine, we must also check if the capacity constraint is violated in the while loop. Now, $\mu$ will still continue to be a valid matching as we only update $\mu$ is we find a leximin increase. This update happens only after checking that $\mu'$ is still a valid matching in the while loop. As a result, $\mu$ is always a valid matching.

Despite the additional capacity constraints, the five observations listed in Section \ref{sec:algs}.1 continue to hold. Thus by analogous reasoning to the uncapacitated setting, we can optimize the leximin tuple, one entry at a time. Thus, CapFaSt correctly finds the leximin optimal stable matching, from a similar argument as Theorem \ref{thm:main}. 

\begin{algorithm}[!ht]
    \small{
     \KwIn{Instance of ranked isometric valuations with capacities $\langle \symbA,\symbB,V, B\rangle$}
     \KwOut{$\mu$}
     Initiate a stable matching: $\mu$ as the \nameA optimal stable matching\;
     Initialize $i\gets n-1$, $down \gets m$ \;
     Initialize $up\gets \max \{j\in [m]| |\mu(\B_j)|>1\}$\;
     Set $\mathcal{L}$ as the leximin tuple for $\mu$\;
     Set $pos[i]$ as the position of $\A_i$ in $\mathcal{L}$, $i\in [n]$\;
     
     Initialize $\mathcal{F}\gets \{\A_n\};\,//$ stores the agents whose matching is fixed.\\
     \While{$i>down-1$ AND $down>1$ }{
     \eIf{$v_{down}(\mu)\geq v_{i(down-1)}$ OR $|\mu(j)|=b_j$}
     {
        $down\gets down-1$;
     }
     {
     
      \eIf{$[v_{i\,down}>v_{down}(\mu)]$}{
        $\mu \gets Demote(\mu,i,up,down)$\;
       }{
      \eIf{$v_{i\,down}<v_{down}(\mu)$}{
        $down\gets down-1$\;
      }{
        // Tie-breaking routine: Need to check if sending $\A_i$ to $\B_j$ will improve the leximin value \;
        $k\gets i-1$\;
        $t\gets pos[i]$\;
        $\mu'\gets Demote(\mu,i,up, down) $\;
        \While{$k>down-1$ AND $|\mu'(down)|\leq b_j$}{
            \eIf{$u_{k\, down}>\mathcal{L}[t] $}{
        $\mu \gets Demote(\mu',k,up, down)$\;
        $i\gets k$\;

        \textbf{break}\;
       }{
      \eIf{$v_{ij}<v_{j}(\mu)$}{
        $down\gets jdown-1$\;
        \textbf{break}\;
      }{
        //We have another tie and send $\A_k$ to $\B_j$ tentatively\;
        $\mu'\gets Demote(\mu',k,up, down)$\;
        $k\gets k-1$, and $t\gets t+1$\;
      }
      
      }
        }
         \If{$k=down-1$ AND $\mu\neq \mu'$}{$down\gets down-1$\;}
        
      }
      }
     }
    
    $\mathcal{F}\gets \{\A_i,\cdots,\,\A_n\}\cup \{\B_{j+1},\cdots \B_m\}$\;
    $Update(\mathcal{L},\mu,pos)$\;
    $i\gets i-1$\;
    
    \If{$(|\mu(\B_{up})|=1$ OR $down=up)$ AND $(up>1)$}{
        $up\gets \max \{j< up\,n |\,\mu(\B_j)|>1\}$
    }
    }
    }
     \caption{CapFaSt}
    \label{alg:fastcap}
    \end{algorithm}
\newpage
\subsection{General Ranked Valuations}
We now translate FaSt-Gen for capacity constrained settings. We again start with the \nameA optimal complete stable matching. This is defined in the preprocessing routine given in Algorithm \ref{alg:cappreproc}. This matches as many \nameAs as possible to $\B_1$ then if there are more than $m-1$ \nameAs unmatched, as many as possible to $\B_2$ and so on. Note that the fixing carried out after defining the matching is not necessary for the correctness of the algorithm, and the algorithm would still continue to correctly compute the leximin optimal stable matching if we did not perform this. In the main algorithm in Algorithm \ref{alg:capleft}, we essentially run the main while loop in FaSt-Gen multiple times. We must do this because a \nameB that is at full capacity may have given out some \nameAs after having the upper limit fixed. In such a case we may see a leximin increase by adding more \nameAs to this \nameB. 

As a result, in Algorithm \ref{alg:capleft}, we unfix the upper limit of all such \nameBs who after having been at full capacity, gave out some \nameAs, till such time that there are no longer any such \nameBs. Since we never add more \nameAs to a \nameB at full capacity (due to the if condition in line \ref{eq:b} of Algorithm \ref{alg:capleft}), $\mu$ continues to be a valid matching.
Now for the look ahead routine detailed in Algorithm \ref{alg:capgenlook}, the only \nameB to which more \nameAs are added is $\B_{down}$ and the condition of the while loop ensures that it is never overfilled. The correctness of the capacitated version of the algorithm follows from the correctness of FaSt-Gen.

One point that needs scrutiny is the running time. While it may appear that the running time may become uncontrolled, note that a particular \nameA and \nameB pair are only ever considered in one particular execution of the first while loop, in which case they can be consider up to $m-1$ times. The comparison of a leximin tuple also continues to take $O(n)$ time. As a result, the running time continues to be $O(m^2n^2)$.  

\begin{algorithm}[!ht]{\small
    \KwIn{Instance of general ranked valuations $\langle \symbA,\symbB,U,V,B\rangle$}
    \KwOut{$\mu$, $UpperFix,\,LowerFix,\,Unfixed$}
    //Initiate a stable matching with the \nameA optimal complete stable matching: \; 
    $j\gets 1$\;
    $p\gets n-m+1$\; $t\gets 0$\;
    \While{$j\leq m$}{
    $i\gets \min \{b_j,p\}$\;
    $\mu(\B_j)\gets \{\A_{t+1},\cdots, \A_{t+i}\}$\;
    $t\gets t+i$\;
    \eIf{$i<p$}{
        $p\gets (p-i)+1$\;
    }{
        $p=1$\;
    }
    $j\gets j+1$\; 
    }
    
    $UpperFix \gets \{\B_1\},\,\, LowerFix \gets \{\B_m\}$\;
    
    Set $\A_i$ to be the highest ranked \nameA s.t. $u_i(\mu)\leq v_j(\mu)$ for all $j\in [m]$, $|\mu(\mu(\A_i))|=1$ and $i\neq 1$ \;
    \If{$i<n$}{$\B_j\gets \mu(\A_i)$\;
    $UpperFix\gets UpperFix \cup \{\B_j,\cdots,\,\B_m\}$\;
    $LowerFix\gets LowerFix \cup \{\B_j,\cdots,\,\B_m\}$\;
    }
    
    $SoftFix \gets \emptyset$\; $Unfixed \gets UpperFix^c$\;
}
\caption{Preprocessing}\label{alg:cappreproc}
\end{algorithm}    

\newpage
\begin{algorithm}[!ht]{\small
    \KwIn{Instance of general ranked valuations with $\langle \symbA,\symbB,U,V, B\rangle$}
    \KwOut{$\mu$}
    $\langle \mu,\, UpperFix,\,LowerFix,\, Unfixed\rangle \gets Preprocessing(\langle\symbA,\symbB,U,V,B\rangle)$\;
    $Flag\gets True$\;
    \While{Flag}{
    Set $T[j]\gets 0$ for all $j \in [m]$\;
    \While{$|UpperFix \backslash LowerFix| + |LowerFix|<m$}{
        $up\gets \min_{j\notin LowerFix } j$\;
        $down \gets \argmin_{j\in Unfixed} v_j(\mu)$.\;
        $SoftFix \gets SoftFix \backslash \{(j,j')| j'\leq up < j \}$\;
        
        \eIf{$|\mu(\B_{up})|=1$ OR $\,v_{up}(\mu)\leq v_{down}(\mu)$}{
            $LowerFix\gets LowerFix \cup \{\B_{up}\}$\;
        }
        {   
        \eIf{$|\mu(\B_{down})|=b_{down}$} 
        { 
            $UpperFix\gets\, UpperFix\cup \{\B_{down}\}$ \label{eq:b}
        }
        {
            $\mu' \gets Demote(\mu,down,up)$\;
            \eIf{$\mathcal{L}_{\mu'}\geq\mathcal{L}_{\mu}$}
            {
                \If{$|\mu(\B_{up})|=b_{up}$}{
                    $T[up]=1$\;
                }
                $\mu \gets \mu'$\;
            }
            {//Decrease in leximin value, need to check the source of the decrease\;
            \eIf{$sourceDec(\mu',\mu)=\B_{up}$}
            { 
                $LowerFix\gets LowerFix \cup \{\B_{up}\}$\;
                $UpperFix\gets UpperFix \cup \{\B_{up+1}\}$
            }{
                \eIf{$sourceDec(\mu',\mu)\in \symbA$}
                {
                    $\B_t\gets \mu(sourceDec(\mu',\mu))$\;
                    $LowerFix\gets LowerFix \cup \{\B_t\}$\;
                    $UpperFix\gets UpperFix \cup \{\B_{t+1}\}$\;
                    $A\gets \{j| j>t+1, j\in Unfixed \}$\;
                    $SoftFix \gets SoftFix \cup (A\times \{t+1\})$\;
                }
                {//The source of the decrease is a \nameB which is still unfixed. Need to check if this may lead to an eventual leximin increase
                    $(\mu,LowerFix,UpperFix,SoftFix) \gets LookaheadRoutine(\mu,down$, $LowerFix$, $UpperFix$, $ SoftFix)$\;
                }
            }
            
            }
        }}
        $Unfixed \gets \{j|j\notin UpperFix $ or $ (j,j')\notin SoftFix$ for some $j'>j\}$\;
        
    }}
    \eIf{$t[j]=0$ for all $j\in [m]$}{
        $Flag\gets False$\;
    }
    {
        Set $j$ to be the highest ranked \nameB such that $T[j]=1$\;
        $UpperFix \gets \{\B_1,\cdots \B_{j-1}\}$\;
        $LowerFix \gets \{\B_m\}$\;
    }
    }
    \caption{CapFaSt-Gen}\label{alg:capleft}
    \end{algorithm}
\newpage    

\begin{algorithm}[!ht]{\small
    \KwIn{$I,\mu,down,LowerFix,UpperFix,SoftFix$}
    \KwOut{$\mu,LowerFix,UpperFix,SoftFix$}
    $\langle \mu', LF, UF \rangle \gets \langle \mu,\, LowerFix,\, UpperFix\rangle $\;
    \While{($|LF|+|UF\backslash LF|<m$) AND ($|\mu(\B_{down})<b_{down}|$)}{
        $up\gets \min_{j\notin LowerFix } j$\;
        \eIf{$|\mu(\B_{up})|=1$ OR $v_{up}(\mu)\leq v_{down}(\mu)$}{
            $LF\gets LF \cup \{\B_{up}\}$\;
        } 
        {
            $\mu' \gets Demote(\mu',up,down)$\;
            \eIf{$\mathcal{L}(\mu')\geq\mathcal{L}(\mu)$}
            {
                $\mu \gets \mu'$\;
                $LoweFix\gets LF$, 
                $UpperFix \gets UF$\;
                \textbf{break}\;
            }
            {//Decrease in leximin value, need to check the source of the decrease\;
            \eIf{$sourceDec(\mu',\mu)=\B_{up}$}
            { 
                $LF\gets LF \cup \{\B_{up}\}$, $UF \gets UF \cup \{\B_{up+1}\}$\;
            }{
                \If{$sourceDec(\mu',\mu)\in \symbA$}
                {
                    $\B_t\gets \mu'(sourceDec(\mu',\mu))$\;
                    \eIf{$t=down$}{
                        //Cannot increase leximin value due to $\B_{down}$\;
                        $UpperFix\gets UpperFix \cup \B_{down}$\;
                    }
                    {
                         $SoftFix \gets SoftFix \cup (down,t)$\;
                    }
                    \textbf{break}\;
                    
                }
                
            }
            
            }
            
            }
    }}
    \caption{Look ahead Routine}\label{alg:capgenlook}
    \end{algorithm}
\end{document}